\documentclass[10pt]{article} 

\usepackage{times}
\usepackage{standalone}
\usepackage{arxiv}

\usepackage{amsmath}
\usepackage{amssymb}
\usepackage{amsthm}

\usepackage{natbib}
\usepackage{hyperref}

  \usepackage{pgfplots}
\usepackage[english]{babel}
\usepackage{color}
\usepackage{enumitem}

\usepackage{float}
\usepackage{pgf}
 \usepackage{tikz}
 \usetikzlibrary{shapes,external}
 \usetikzlibrary{pgfplots.groupplots}
\tikzset{naming/.style={align=center,font=\footnotesize}}
\tikzset{area/.style = {draw, shape = regular polygon, regular polygon sides = 6, thick, minimum width = 5cm}}
 \usepackage{subfigure}
 \usepackage{caption}

\usepackage[noend]{algorithm2e}
\usepackage{etoolbox}
\AtBeginEnvironment{algorithm}{\SetInd{.5em}{1em}}
\SetKwComment{Comment}{/* }{ */}
\RestyleAlgo{ruled}

\usepackage{multirow}
 \usepackage{calrsfs}
\usepackage[mathscr]{eucal}
\usepackage{subfigure}
\usepackage{todonotes}
\usepackage{upgreek}
\usepackage{float}
 
\usepackage{comment}

\renewcommand{\qed}{\hfill $\blacksquare$}
\newcommand{\R}{\mathbb R}
\newcommand{\N}{\mathbb N}

\renewcommand{\bar}{\overline}
\newcommand{\what}{\widehat}
\newcommand{\eq}{\mathcal C}

 \theoremstyle{plain} \newtheorem{theorem}{Theorem}[section]
 \newtheorem{remark}{Remark}[section] \newtheorem{definition} {Definition} [section]
 \newtheorem{lemma} {Lemma} [section] \newtheorem{corollary} {Corollary} [section]
  \theoremstyle{definition}
 \newtheorem{proposition}{Proposition}[section] \newtheorem{example}{Example}[section]

\newtheorem{observation}{Observation}

\usepackage{enumitem}

\newcommand{\floor}[1]{\left\lfloor #1 \right\rfloor}

\DeclareMathOperator{\id}{Id}

\DeclareMathOperator{\argmin}{argmin}

\DeclareMathOperator{\relint}{relint}

\DeclareMathOperator{\Proj}{Proj}

\DeclareMathOperator{\cA}{\mathscr{A}}
\DeclareMathOperator{\cB}{\mathscr{B}}
\DeclareMathOperator{\cC}{\mathscr{C}}

\DeclareMathOperator{\cF}{\mathscr{F}}

\DeclareMathOperator{\cL}{\mathscr{L}}

\DeclareMathOperator{\cP}{\mathscr{P}}

\DeclareMathOperator{\cR}{\mathscr{R}}
\DeclareMathOperator{\cS}{\mathscr{S}}

\DeclareMathOperator{\cU}{\mathscr{U}}

\DeclareMathOperator{\cX}{\mathscr{X}}

\DeclareMathOperator{\cZ}{\mathscr{Z}}

\DeclareMathOperator{\NN}{\mathbb{N}}

\DeclareMathOperator{\RR}{\mathbb{R}}

\renewcommand{\leq}{\leqslant}
\renewcommand{\geq}{\geqslant}

\newcommand{\norm}[1]{\left\| #1 \right \|}
\newcommand{\tnorm}[1]{\| #1  \|}
\newcommand{\absv}[1]{\left| #1  \right|}

\newcommand{\dprod}[1]{\left\langle #1 \right\rangle}

\DeclareMathOperator{\dynReg}{\text{dReg}}

\DeclareMathOperator{\track}{\tau}

\DeclareMathOperator{\muMono}{\mu}

\usepackage{natbib}
\usepackage[symbol]{footmisc}

\let\top\intercal 

\newcommand{\shorttitle}{Time-Varying Variational Inequalities}
\title{Tracking Solutions of Time-Varying Variational Inequalities}

\usepackage{authblk}

\author[1]{Hedi Hadiji$\phantom{}^*$}
\author[2]{Sarah Sachs$\phantom{}^*$}
\author[3]{Cristóbal Guzmán}

\affil[1]{\centering Laboratoire des signaux et systèmes,
Univ. Paris-Saclay, CNRS, CentraleSupélec } 
\affil[2]{\centering School of Mathematics,
University of  Bristol }
\affil[3]{\centering Pontificia Universidad Católica de Chile,
Institute for Mathematical and Computational Engineering,
Facultad de Matemáticas and Escuela de Ingeniería 
\vspace{1em}
}

\affil[*]{equal contribution}

\date{\vspace{-1em}\today}

\begin{document}
 \maketitle

 \begin{abstract}{%
Tracking the solution of time-varying variational inequalities is an important problem
with applications in game theory, optimization, and machine learning. Existing work
considers time-varying games or time-varying optimization problems. For strongly convex
optimization problems or strongly monotone games, such results provide tracking
guarantees under the assumption that the variation of the time-varying problem is
restrained, that is, problems with a sublinear solution path. We extend
existing results in two ways: In our first result, we provide tracking bounds for (1)
variational inequalities with a sublinear solution path but not necessarily monotone
functions, and (2) for periodic time-varying variational inequalities that do not
necessarily have a sublinear solution path length. Our second main contribution is an
extensive study of the convergence behavior and trajectory of discrete dynamical systems
of periodic time-varying VI. We show that these systems can either exhibit
provably chaotic behavior or can converge to the solution.
 }\end{abstract}%

\section{Introduction and setting.}\label{sec:introduction}
Solutions of variational inequalities generalize minimizers of functions, saddle-points,
fixed points, and many notions of equilibria across various disciplines. Practitioners
often encounter problems in which the parameters of the variational inequality have a
component that depends on time. We study the problem of tracking the sequence of
solutions in time-varying variational inequalities in an online setting: that is, sequentially predicting the solutions of a sequence of time-varying variational inequality problems.

 \paragraph{Formal Setting}
Let $\cZ$ be a closed convex subset of $\R^d$ for $d\geq1$. We let $\cC(\cZ, \R^d)$ denote the set of all
continuous functions $\cZ \to \R^d $. The finite-dimensional \emph{Variational
Inequality Problem} (VI problem) for the operator $F$ on the
domain $\cZ$, referred to as $\mathrm{VIP}(F,\cZ)$ from now on, is defined as:
Given $F \in \cC(\cZ, \R^d)$,
\begin{align} 
  \label{VIP}
  \text{find \, $Z^\star \in \cZ$ \, such that} \qquad
\dprod{F(Z^\star),Z-Z^\star} \geq 0 \quad \text{for all} \quad Z \in \cZ. \tag{VIP}
\end{align}
 {Let $\norm{\, \cdot\,}$ denote any norm in $\RR^d$.} For $\mu > 0$, the operator $F$ is said to be $\muMono$-\emph{strongly monotone} if
 \begin{align}
 \forall Z,Z' \in \cZ: \quad \dprod{F(Z) - F(Z'), Z-Z'} \geq \muMono\norm{Z-Z'}^2.
 \label{sMono}\tag{$\mu$-Mono}
 \end{align}
The operator $F$ is \emph{monotone} if the above holds for $\muMono= 0$.
Given $L\geq 0$, we say $F$ is $L$-\emph{Lipschitz} if
 \begin{align}
 \forall Z,Z'\in \cZ:\quad \|F(Z)-F(Z')\|\leq L\|Z-Z'\|
 \label{Lip} \tag{$L$-Lip}
 \end{align}
 Many common problems in mathematics can be formulated as VI problems.  Prime
 examples are convex optimization and Nash equilibrium computation in concave games.
 See e.g., \citet[Chapter~12]{Rockafellar:1998aa}, \citet[Chapter
 7]{Kinderlehrer:2000aa},  { or \citet{facchinei2003finite-d}} for more details and examples. Beyond these well-known cases,
 generalized linear models (GLMs) \cite{Juditsky:2019} also fit within the framework of
 monotone VI problems, though they are neither convex optimization problems nor concave
 games but statistical parameter estimation problems (see Example \ref{Example:GLM} for a time-varying illustration). In Section
 \ref{sec:examples}, we discuss several examples in detail and the guarantees our results
 can provide for these application examples.
 
We study the problem of online tracking of solutions for time-varying VI. More
precisely, we consider a sequence of operators $(F_t)_{t \in \NN}$ revealed in an online
fashion to a learner. At time step~$t$, given the past observations, the learner outputs
a candidate solution $Z_t$, and then observes the operator $F_t$. Formally, an online
algorithm is a sequence of update rules $\cA_t:\cZ \times \cF^{t-1} \rightarrow \cZ$
which act online: at time step $t$ the output of the algorithm may depend on
$F_{1}, \dots F_{t-1},Z_{t-1}$, but not on $(F_s)_{s\geq t}$. An instance of a
\emph{time-varying VI problem} is characterized by the sequence $(F_t)$ together with
the domain $\cZ$; we denote
it by $\mathrm{VIP}((F_t)_{t \in [T]}, \cZ)$.

We focus especially on sequences of operators that all admit a unique solution, denoted by
$(Z_t^\star)$.
A sufficient, but not necessary, condition for uniqueness is that the operator be
strongly monotone. For  $\mathrm{VIP}((F_t)_{t \in [T]}, \cZ)$, we examine the \emph{tracking guarantees} of
algorithms. Denoting the $t$-th output of algorithm $\cA$ by $Z_{t} =
\cA_t(Z_{t-1},(F_{s})_{s \in [t-1]})$, the \emph{tracking error} of $\cA$ is 
\begin{align*}
  \track_T(\cA) = \sum_{t=1}^T \norm{Z_t - Z^\star_t}^2.
\end{align*}
We exhibit conditions under which this tracking error is sublinear, i.e., there
exists an $\alpha \in[0,1)$ and constant $c\geq0$, such that $\track_T(\cA) \leq c
T^{\alpha}$. Note that guaranteeing a sublinear tracking error for time-varying VI
problems may be impossible for some sequences, especially when solutions vary too much,
 see Lemma~\ref{LB:tracking} in Appendix~\ref{sec:appendixIntro}. 
 This motivates the introduction of the quadratic path length, which measures the hardness
of time-varying VI problems
\[
  P^\star_T = \sum_{t=2}^T\norm{Z_t^\star - Z_{t-1}^\star}^2 \,.
\]
Since general results for arbitrary sequences are not achievable, we restrict our
attention to two important special cases defined below: tame and periodic time-varying VI.
\smallskip
\begin{definition}[Tame Time-Varying VI problem]
  \label{def:tameVIP}
We call a time-varying VI problem \emph{tame} if the quadratic path length of the solutions
$P^\star_T$ is sublinear. That is, if there exists time-independent constants $\alpha \in
[0,1)$ and
$c>0$ such that $P_T^\star \leq c T^\alpha$ for all $T \geq 2$. 
\end{definition}
\smallskip

For tame time-varying VI problems, we show tracking guarantees for the class of algorithms 
satisfying a one-step contraction property, in Section~\ref{sec:contraction}. In
Sections~\ref{sec:periodic}
and~\ref{sec:chaos}, we focus on periodic time-varying problems, which we define
naturally as follows. Note that periodic problems cannot be tame unless the equilibrium
is constant.
 
\smallskip
\begin{definition}[Time-varying VI problem with periodic solutions] \label{def:periodicVIP}
A  time-varying VI problem $\mathrm{VIP}((F_t),\cZ)$ has \emph{periodic solutions} if there exists $k \in \NN$,
such that $Z_{t+k}^\star = Z_t^\star$ for all $t \in \NN$.
\end{definition}
\smallskip

Periodic time-varying VI problems, that is, time-varying VI problems with $F_{t+k} = F_t$ for all $t
\in \NN$, are time-varying VI problems with periodic solutions; the converse does not necessarily hold.

\paragraph{Summary of contributions.}
 {
We provide three sets of contributions to the study of tracking in time-varying VI, also
summarized in Table~\ref{table:results}:
\begin{enumerate}
  \item Section~\ref{sec:contraction} focuses on \emph{tame} time-varying VI problems. We
    provide path-length bounds for algorithms satisfying a one-step contraction property.
    This condition is satisfied in a variety of settings, with bounded or unbounded domains,
    strongly monotone and non-strongly-monotone cases; see
    Section~\ref{sec:examplesContraction} for examples. 
  \item We then shift our attention to \emph{periodic} and \emph{strongly-monotone}
    sequences of operators, with two main results. We first give an algorithm with a
    logarithmic tracking bound for \emph{bounded domains} and \emph{bounded operators}, in
    Theorem~\ref{thm:aggregation}.
    In the case of \emph{Lipschitz operators on unbounded domains}, we design an algorithm
    with a constant tracking bound in Theorem~\ref{thm:constRegretAdaHedge}.
    Both algorithms rely on a meta-algorithm aggregating over a collection of algorithms.
    Hence, they adapt to the a priori unknown period of the sequences and only
    require an upper bound on that period. 
  \item  {In Section~\ref{sec:chaos}, we examine gradient descent with a constant step size on a periodic optimization problem. We show that the composition of the periodic operators yields an autonomous (time-invariant) dynamical system. In this autonomous setting, we observe a range of behaviors, including convergence, periodicity, and Li-Yorke chaos, depending on the step size.}
\end{enumerate}
{We stress that the emergence of chaos, as discussed in Section~\ref{sec:chaos}, is a property of the autonomous dynamical system induced by the periodic composition of operators, and not of general non-autonomous (time-varying) problems.}
We illustrate the practical consequences of our tracking guarantees with examples in Section \ref{sec:examples}.}

 {
\begin{remark}[Ignoring Time-Variations and Practical Relevance]
  In the context of time-varying problems with uncertainty about their structure, a first
  (arguably reasonable) attempt would be to
  ignore the problem's time-varying nature. If fluctuations are
  small, their impact on solutions may be negligible. 
   
  With this in mind, we examine a basic example of gradient descent with a fixed
  step-size applied on a periodic sequence of strongly convex functions with a shared
  minimizer (see Appendix~\ref{app:example_simple_periodic}),
  and show that tuning the learning rate for convergence is hard despite the simplicity
  of the example.
  {Section~\ref{sec:chaos} furthers this investigation, showing chaotic phenomena even in the case where all periodic functions share the same minimizer.     These negative results in such simple settings justify and motivate the
   more involved approach carried out in Section \ref{sec:periodic}.   }
 \end{remark}
}

\begin{center}
\begin{table}[t]
\begin{tabular}{|p{1.8cm}|@{}c@{}|}\hline
& 
\begin{tabular}{p{4.8cm}|p{8.9cm}|}
  Algorithm $\cA$ & Results \\
\end{tabular}
\tabularnewline\hline
 \vspace{-1em}
Tame
\ref{VIP}
&
\begin{tabular}{p{4.8cm}|p{8.9cm}|}
  \smallskip
  Any contractive algorithm, e.g.:
  \begin{itemize}
    \item[-] Gradient descent
    \item[-] Resolvent iteration
  \end{itemize}
   & 
  \vspace{-0.5em}
   \textbf{Theorem}~\ref{thm:contraction2track}: Sublinear tracking error  
 \vspace*{-0.5em}
   \[
   \track_T(\cA) \leq  cT^\alpha\,.
 \vspace*{-0.5em}
   \]
  \textbf{Assumptions:} tame (Def.~\ref{def:tameVIP}), contractive (Def.~\ref{def:contraction}).
   {\textbf{Remark:} strong monotonicity not necessarily required (cf. Example \ref{Example:nonmonoGame}).}
\end{tabular}
\tabularnewline \hline
$k$-Periodic \ref{VIP}
&
\begin{tabular}{p{4.8cm}|p{8.9cm}|}
\smallskip
Meta-algorithm with: 
\begin{itemize}[leftmargin=0.1cm]
\item[]{Aggregation:} EW, constant learning rate
\item[]{Base Algorithm:} OGD
\end{itemize}  
&  
\vspace{-0.5em}
\textbf{Theorem}~\ref{thm:aggregation}: Logarithmic tracking error
 \vspace*{-0.5em}
 \[
 \track_T(\cA) \leq c \frac{k(G + D)^2}{\mu^2}\log T + c'.
 \vspace*{-0.5em}
\]
 \textbf{Assumptions:} $k$-periodic (Def.~\ref{def:periodicVIP}), bounded domain, strongly monotone (\ref{sMono}) and bounded operators
 \newline
  {\textbf{Remark: }Computational cost per iteration is constant.} \\
 \hline
\smallskip
Meta-algorithm with: 
\begin{itemize}[leftmargin=0.1cm]
\item[]{Aggregation:} EW, adaptive learning rate
\item[]{Base Algorithm:} GD
\end{itemize}  
&
\vspace{-0.5em}
 \textbf{Theorem}~\ref{thm:constRegretAdaHedge}: Constant tracking error
 \vspace*{-0.5em}
\[ 
  \track_T(\cA) \leq c \, k D_0^2\frac{L^4}{\mu^4}    \log K + c'.
 \vspace*{-0.5em}
\]
\textbf{Assumptions:} 
$k$-periodic (Def.~\ref{def:periodicVIP}),  $L$-Lipschitz,  
$\mu$-strongly monotone operators.
\newline 
 {\textbf{Remark:} Domain can be unbounded. Computational cost per iteration is in the order of $K$. }
\\
 \hline
\bigskip
\smallskip
GD, constant step-size &   
 \begin{itemize}[leftmargin=2cm]
\item[\text{Phase I\phantom{II}:}] Convergence. 
\item[\text{Phase II\phantom{I}:}] Composite phase with convergence, periodicity, chaos, and divergence. (Obs.~\ref{obs:gd_behaviours})
\item[\text{Phase III:}] Divergence. 
\end{itemize} 
\\
\end{tabular}
\tabularnewline\hline
 \end{tabular}
 \caption{Overview of Results. Everywhere, $c, c'$ denote constants independent of the
   number of iterations  $T$. The constant $D_0 = \max_{i \in[k]}\tnorm{Z_1 - Z_i^\star}$
   denotes the maximal distance of the initial iterate to the solutions $Z_1^\star, \dots,
   Z_k^\star$, and $K$ denotes a known upper bound on $k$.  $D_0$ and $K$ are instance
   dependent constants independent of $T$. 
The abbreviations of the algorithms are: EW: Exponential Weights, OGD: Online Gradient
descent, GD: Gradient Descent. For details on the algorithms, see Section
\ref{sec:periodic}. } \label{table:results}
  \end{table}
\end{center}

\subsection{Dynamic regret minimization. } Note that showing tracking bounds for an algorithm is related to
\emph{dynamic regret minimization} \citep{zinkevich2003online} in the special case of time-varying optimization problems. Indeed,
consider a sequence of functions $f_1, \dots f_T$ and a sequence of comparators $ u_1,
\dots, u_T \in \cZ$. These comparators can be the minimizers $u_t^\star$ of
the functions $f_t$. The dynamic regret of an online algorithm $\cA$ is
\begin{align*}
  \dynReg^{\cA}_T\left((u_t)_{t \in [T]}\right) := \sum_{t=1}^T \Big( f_t(x_t) - f_t(u_t)\Big). 
  \end{align*}
  Note that the comparators $(u_t)$ are not necessarily the minimizers of
the functions $(f_t)$. Guarantees for dynamic regret where the comparators
are restricted to the minimizers $u_t^\star$ of strongly convex and smooth functions
$(f_t)$ have been studied, e.g. in \cite{Zhao:2020aa}.
  Relating these results to tracking, we note that the quantity $( 2 / \mu) \dynReg_T^{\cA}\left((u^\star_t)_{t \in [T]}\right)$ is an
upper bound on the tracking error $\track_T(\cA)$ over the set of $\mu$-strongly convex
functions $\hat \cF_{\mu}$.
 Hence,
if an algorithm has sublinear dynamic regret, it is a tracking algorithm. Dynamic regret
bounds typically depend on the \emph{path length} $P^{1}_T = \sum_{t=1}^{T-1}
\tnorm{u_t - u_{t+1}}$, that is, the norm differences of the comparator terms.  
Indeed,
existing lower bounds on the dynamic regret show that $\sqrt{T(P^{1}_T + 1)}$ is
unavoidable for convex functions
\cite{zhang2018adaptive}. Hence, showing sublinear
bounds require additional assumptions, either on the path length $P^1_T$ or on the
distances between the functions $f_t$.

\subsection{Related work.}
 Tracking equilibria in time-varying games has received growing interest in the game theory community in recent years,  
see~\citet{rivera-cardos2019competin}, \citet{fiez2021online}, \citet{Mertikopoulos:2021aa}, \citet{Duvocelle_2023}, or  \citet{Feng:2023aa}. 
 \citet[Theorem 2 and Corollary 3]{Duvocelle_2023} provided guarantees for tracking
equilibria in time-varying games under a very general feedback model that allows for
noisy and biased observations. In the tame case, they provide tracking bounds for a
proximal point method under strong monotonicity. \citet{Yan:2023aa}
strengthen these results in the exact feedback model under an additional smoothness
assumption, providing tighter and more adaptive bounds. \citet{Feng:2023aa} point out a
surprising example in time-varying bilinear games, in which the extra-gradient method
outperforms optimistic methods of \citet{rakhlin2013online}; these two algorithms were
often thought to enjoy almost identical behavior. 

A related line of work studies no-regret dynamics in games, focusing on algorithms
designed to minimize regret at the individual player level. 
In repeated static strongly-convex-strongly-concave zero-sum games, the sum of the regret
of the players bounds the distance of the joint average plays to the equilibrium of the
game (see, e.g., Example~\ref{ex:gap_saddle_point}). \citet{zhang2022no-regret} investigates the impact of time variations of the
game, and derives tracking guarantees in the tame case when players use dynamic regret
algorithms.

As seen in the previous section, tracking solutions in the special case of optimization are related to dynamic regret
minimization. The literature distinguishes dynamic regret against arbitrary comparators
(\citet{zinkevich2003online, baby2022optimal}) and against the minimizers of the
sequences of functions observed (\citet{besbes2015non-stat}); the former is a more
demanding objective, while the latter suffices for tracking when functions are strongly convex. In this case, the dynamic regret against minimizers bounds the tracking error up to a multiplicative factor. For
strongly convex and smooth functions, \citet{Mokhtari2016Online-o}, using the fact that
gradient descent is contractive (see Definition~\ref{def:contraction}), provides tracking
bounds of order $P_T^\star$. Our analysis in Theorem~\ref{thm:contraction2track} is a
generalization of their results to VI problems and to other algorithms. 
{For an in-depth comparison, see Section~\ref{subsec:consequences}.}

 Online periodic optimizations and periodic games are less explored topics. 
\citet{fiez2021online} show that the continuous-time dynamics of gradient-descent-ascent
are Poincaré recurrent in periodic bilinear unconstrained zero-sum games, showing a form 
of stability despite the time variations---note that the equilibrium is constant and
equal to $0$ in their analysis. In that same setting, but with discrete time,
\citet{Feng:2023aa} show the convergence of extra-gradient methods to the constant
equilibrium. Our results are complementary: we consider the strongly monotone setting,
and provide tracking bounds. 

The dynamics of optimization methods with large step sizes have received attention
recently, with some works exhibiting chaotic behavior. Chaos may also appear in
natural game dynamics, e.g.\,in \citet{Piliouras_2023,Falniowski:2024aa}. 
Closest to our results of Section~\ref{sec:chaos} is that of \citet{Chen2023FromST}, which
characterize different phases in the behavior of gradient descent in quadratic
regression, as the learning rate increases. To the best of our knowledge, we are the
first to exhibit such a phenomenon in the periodic time-varying case. We also note that
periodic time-varying VI problems are related to incremental gradient descent
\cite{Nedic:2001aa}. Our observations might also be of interest in this context.

 {
  Furthermore, two related lines of research in the area of game theory are worth noting.
  Mechanism design \citep{baudin2023strategicbehaviornoregretlearning, deng2021non} and
  often specifically auction design \citep{nedelec2022learning, rahme2021auction}
  frameworks, in which a central entity modifies the parameters of a game, induce
  time-varying games for the players.
  Stochastic games model repeated settings in which the joint actions of players 
  influence future payoffs, causing time variations in the individual objectives.    Our results may formally apply to some specific settings in these models, but they ignore the crucial strategic  aspects involved in planning for the long-term consequences of current actions. 
  }
   In Section~\ref{sec:chaosGD}, we also mention a relation to the theory of Iterated
Function Systems (IFS), see e.g., \citet[Chapter~9]{Falconer:1990aa}. Relations between
stochastic gradient descent methods and IFS were recently used in
\citet{Hodgkinson:2021aa} to establish convergence or generalization guarantees for
popular machine learning methods. We build on similar relations to establish a link between existing results for IFS and our experimental observations. 

\section{Tracking by contractive algorithms.} \label{sec:contraction}
\phantom{x}
 We derive tracking bounds for time-varying variational inequalities using a contraction property of algorithms. While this property holds for many common first-order algorithms for strongly monotone VIs or VIs with bounded sets $\cZ$, we also show that strong monotonicity or boundedness is not necessarily required for contractiveness. In Section~\ref{subsec:contraction}, we define \emph{contractive} algorithms and show that they automatically enjoy tracking bounds
(Theorem~\ref{thm:contraction2track}). Examples of such algorithms are presented in Section~\ref{sec:examplesContraction}, followed by a discussion of their implications in Section~\ref{subsec:consequences}.

\subsection{Contractive implies tracking.}\label{subsec:contraction}
In strongly convex and smooth optimization, the gradient descent updates ensure that the
distance to the minimizer decreases by a constant multiplicative factor at every
time step. \citet{Mokhtari2016Online-o} use this property to derive dynamic regret bounds
for gradient descent in online optimization. We generalize this observation to
time-varying VI problem, and contractive algorithm, which are defined as follows: 

\smallskip
\begin{definition}[Contractive Algorithms] \label{def:contraction}
Let $C \in (0, 1)$ and $\cA$ be an algorithm composed of update rules $(\cA_t)$. The
update rule $\cA_t$ is said to be $C$-contractive over a set of operators $\cF$ if
for all $F \in \cF$ with solution $Z^\star \in \cZ$, for any $Z \in \cZ$,
\[
    \|Z_+ - Z^\star\| \leq C  \|Z - Z^\star\| \, , 
\]
where we denoted by $Z_+$ the image of $Z$ under the update rule $\cA_t$ after observing
$F$. 
We say that the algorithm $\cA$ is $C$-contractive over $ \cF$ if $\cA_t$ is
$C$-contractive for all $t\geq1$. 
\end{definition}

\begin{theorem}   \label{thm:contraction2track}
Suppose $\cA$ is $C$-contractive over $\cF$. Then for any sequence of operators in
$\cF$ with solutions $(Z_t^\star)_{t \in [T]}$, the tracking error is bounded by
  \begin{align*}
      \track_T(\cA) \leq \frac{1}{(1 -  C)^2} \sum_{t=2}^T \|Z_t^\star - Z_{t-1}^\star\|^2 
    + \frac{1}{1 -  C}\|Z_1 - Z_1^\star\|^2 \, .
  \end{align*}
\end{theorem}
\begin{proof}{Proof.}
Using the inequality $\|u + v\|^2 \leq (1 + \alpha)\|u\|^2 + ( 1+ \alpha^{-1}) \|v\|^2$
for any $u, v \in \R^d$ and $\alpha > 0$,
\begin{align*}
\norm{Z_{t+1} - Z^\star_{t+1}}^2
  &\leq \frac{1}{C} \norm{Z_{t+1} - Z^\star_{t}}^2
  + \Big(1 + \frac{1}{1 / C - 1} \Big)\norm{Z_{t+1}^\star - Z_t^\star}^2 \\
  &\leq  C \norm{Z_t - Z_t^\star}^2 + \frac{1}{1 - C} \norm{Z^\star_{t+1} - Z^\star_t}^2 \,,
\end{align*}
 where the second inequality follows by the contraction property. Summing over
$t \in [T-1]$ yields
  \[
    \sum_{t=1}^{T-1}  \|Z_{t+1} - Z^\star_{t+1}\|^2
    \leq C \sum_{t=1}^{T-1} \|Z_t - Z^\star_{t}\|^2
    + \frac{1}{1 - C} \sum_{t=1}^{T-1} \|Z_{t+1}^\star - Z_t^\star\|^2\, .
  \]
  Then, adding $\|Z_1 - Z_1^\star\|^2$ to both sides, and upper bounding the
  first sum on the right-hand side,
  \[
    \sum_{t=1}^{T}  \|Z_{t} - Z^\star_{t}\|^2
    \leq C \sum_{t=1}^{T} \|Z_t - Z^\star_{t}\|^2
    + \frac{1}{1 - C} \sum_{t=1}^{T-1} 
    \|Z_{t+1}^\star - Z_t^\star\|^2 + \|Z_1 - Z_1^\star\|^2 \, .
  \]
  Therefore, after reorganizing the sum,
    \[
    \sum_{t=1}^T \|Z_t - Z_{t}^\star\|^2
    \leq \frac{1}{(1 - C)^2} \sum_{t=1}^{T-1}
    \|Z_{t+1}^\star - Z_{t}^\star\|^2 + \frac{1}{1 - C}\|Z_1 - Z_1^\star\|^2 \, , 
  \]
  which completes the proof. 
\qed
\end{proof}
 Note that Theorem \ref{thm:contraction2track} does not require $\cZ$ to be bounded. However, some algorithms require a bounded convex set $\cZ$ in order to satisfy the contraction property (Definition \ref{def:contraction}), e.g., Example \ref{example:EGD}. 
 
\smallskip

\subsection{Examples of contractive algorithms.}\label{sec:examplesContraction}
The contraction property (Definition~\ref{def:contraction}) is satisfied in a wide range
of settings, often in the presence of curvature of the objective (in the form of strong
monotonicity of the operator), but not always. Specifically, Example
\ref{Example:nonmonoGame} shows that strong monotonicity is not a necessary assumption
for contractivity. Furthermore, note that some examples do not require a bounded
$\cZ$.
\smallskip
\begin{example}\label{example:GD}
Let $F:\cZ \to \R^d$ be a $\mu$-strongly-monotone and $L$-Lipschitz operator, where
$\cZ\subseteq \R^d$ is a closed convex set.
Let $T_\eta : z \mapsto \Proj_{\cZ}(z - \eta F(z) )$ denote the projected forward
operator 
for some step-size $\eta > 0$.
Then $Z^\star$ is a solution to $\mathrm{VIP}(F, \cZ)$ if and only if it is a fixed point
of~$T_\eta$ (see\,\citet[Proposition~23.38]{Bauschke_2011}). Moreover, for $\eta = \mu/L^2$, the operator $T_\eta$ is a contraction:  
for any $Z_1, Z_2 \in \cZ$, using the nonexpansiveness of Euclidean projections,
strong-monotonicity and the Lipschitz property of $F$,    
 \begin{align*}
\norm{T_\eta(Z_1) - T_\eta(Z_2)}^2 
&\leq \norm{Z_1 - \eta F(Z_1) - (Z_2 - \eta F(Z_2)) }^2 \\
&= \norm{Z_1 - Z_2}^2 - 2\eta \dprod{ Z_1- Z_2 ,F(Z_1) - F(Z_2)}
+ \eta^2 \norm{F(Z_1) - F(Z_2)}^2 \\
&\leq ( 1+\eta^2L^2 - 2\eta \mu ) \norm{Z_1 - Z_2}^2 = (1-\kappa^{-2}) \norm{Z_1 - Z_2}^2 \,,
\end{align*}
where $\kappa = \mu/L$ denotes the condition number. 
This implies that the projected forward algorithm with step-size $\eta = \mu / L^2$ is a
contractive algorithm. In particular, this includes gradient descent for optimization
and gradient-descent-ascent for saddle-point optimization.
\end{example}

\smallskip

\begin{example}\label{example:Hilbert}
Suppose $F: \R^d \rightarrow \R^d$ is a maximally $\mu$-strongly monotone operator.
Let $R$ denote the resolvent operator, that is, $R = (\id + F)^{-1}$. (Note that $R$ is related to the backward iteration, and that operator $R$
coincides with the proximal mapping when $F$ is a gradient operator.)
Then a point $Z^\star$ is a solution to the  $\mathrm{VIP}(F, \R^d)$ if and only if 
$F(Z^\star) = 0$, or equivalently if and only if $R(Z^\star) = Z^\star$. 
By Theorem 2.1. (xi) in~\citet{Bauschke:2011aa}  the resolvent operator is $1/(1+\mu)$-contractive
 This implies in particular that for any $Z\in \R^d$, 
\[
\| R(Z) - Z^\star\| = \| R(Z) - R(Z^\star)\| \leq (1+\mu)^{-1} \|Z - Z^\star\| \,.
\]
\end{example}
\smallskip
\begin{example}\label{Example:nonmonoGame} Consider the following zero-sum game:
\begin{align}
  \label{nonmonoGame}
  \min_{y \in \RR}
  \max_{x \in \RR} \;
  \Big\{ x^2 + 3 \sin^2(x) + a \sin^2(x) \sin^2(y) - y^2 - 3 \sin^2(y)
  \Big\}
  \, ,
\end{align}
where we assume $a \in [0,1]$. Note that the game in \eqref{nonmonoGame} is not a strongly monotone game and it does not
satisfy the diagonal concavity assumption of \citet{rosen1965existenc} (see
Proposition~\ref{prop:nonConvex} in Appendix~\ref{appendix:generality} for details). However, note that for any $y \in \RR$,
the function
$\nu(\, \cdot\,,y) = x^2 + 3 \sin^2(x) + a \sin^2(x) \sin^2(y) - y^2 - 3 \sin^2(y)$
satisfies the so called \emph{restricted secant inequality} (RSI) (see Proposition~\ref{prop:RSI}, Appendix~\ref{appendix:generality} for details and definition). We show
in Proposition~\ref{prop:contraction} in Appendix~\ref{appendix:generality}, that
contraction is satisfied for gradient descent-ascent. A sequence $(a_t)$, with $a_t \in [0,1]$ defines a time-varying instance of this zero-sum game. Note that this example is inspired
and heavily builds on well-known results in optimization by \citet{Karimi2016LinearCO}.
\end{example}

\smallskip

\begin{example} \label{example:EGD}
We conclude these examples with a more informal one, regarding the extra-gradient 
method on bi-affine saddle-point problems, which include in particular equilibria of 
matrix zero-sum games.
Suppose that $\cZ$ is the product of two polytopes $\cP_1, \cP_2$ and that $F$ is a map of the form 
$F(x, y) = c^\top x + d^\top y - x^\top Ay $; that is, $F$ is \emph{bi-affine}.  
\citet{korpelevich1976extragradient} introduced the extra-gradient method for
saddle-point optimization, which consists in playing 
\[
Z_{t+1} = 
\Proj_{\cZ} \big\{
  Z_t - \eta F\big( \Proj_{\cZ} \big\{ Z_t - \eta F(Z_t) \big\} \big)
\big\} \,.
\]
She showed (in Theorem~3 in the reference) that if $F$ has a unique solution $Z_F^\star$, then the
extra-gradient method is eventually a contraction. Precisely, there exists a (relative)
neighborhood of the solution, denoted by $\cU(Z_F^\star)$, such that extra-gradient
becomes a contraction as soon as it reaches $\cU(Z_F^\star)$. Moreover, suppose in
addition that the matrix $A$ is a square matrix and the equilibrium lies in the relative
interior of both $\cP_1$ and $\cP_2$, i.e., in $\relint(\cP_1)\times\relint(\cP_2)$. In
that case, the contraction constant depends only on upper bounds of the euclidean operator
norms of $A$, and of $A^{-1}$.
To fully characterize the contraction property and complete this example, one would need
to explicitly describe $\cU(Z_F^\star)$. Then, given a subset $\cC \subset \cZ$, the
extra-gradient method started at $Z_0\in \cC$ is a contractive algorithm on the set of VI
problems $\mathrm{VIP}(F,\cZ)$ such that $\cU(Z_F^\star)\subset\cC$. Giving an exact
description of $\cU(Z_F^\star)$ is technically challenging and goes beyond the scope of
this article. Recent works (e.g., \citet{wei2021linear}) provide finite-time linear
convergence bounds with explicit constants, but use arguments that do not explicitly rely
upon the contraction property.
\end{example}
\subsection{Tightness of analysis.}
We also show that our analysis of tracking in Theorem~\ref{thm:contraction2track} is
tight, in the sense that one cannot significantly improve the tracking bound uniformly
for all contractive algorithms. The counter-example is built on a family of quadratic
problems and gradient descent. See Appendix~\ref{app:tightness} for a proof.
\begin{theorem}
  \label{thm:tightness}
  For any $C \in (0,1)$, there exists a sequence of VI problems and a $C$-contractive 
  algorithm over this sequence such that
  \[
    \sum_{t=1}^T \|Z_t - Z^\star_t\|^2
    \geq \frac{1}{(1 - C)^2} \sum_{t= 2}^T \|Z^\star_{t} - Z^\star_{t-1}\|^2   \, .
  \]
\end{theorem}

\subsection{Tracking guarantees and comparison to existing results.}
\label{subsec:consequences}
If a time-varying VI problem is $\alpha$-tame (Definition~\ref{def:tameVIP}), then
Theorem \ref{thm:contraction2track} implies that for a $C$-contractive algorithm $\cA$ the tracking error is less than 
\begin{align*}
   \track_T(\cA)  \leq \frac{c}{(1-C)^2} T^{\alpha} + \frac{\norm{Z_1-Z_1^\star}^2}{1-C},
\end{align*}
where $c>0$ is the constant in the $\alpha$-tame property.

To our knowledge, our precise setting (exact feedback for time-varying VIs) has not been 
studied previously at this level of generality.
Previous work on time-varying problems has focused either on optimization or games, 
sometimes in settings with stochastic feedback, in which the learner only gets to
observe a noisy value of the operator instead of the true operator.
Despite the differences in assumptions, let us provide more detailed points of
comparison to the related literature, starting with exact feedback.

Our bound recovers the rates of \citet{Mokhtari2016Online-o} when specialized to
strongly-convex time-varying optimization. These rates are optimal, in the sense that
one cannot improve the dependence on $\alpha$ uniformly over all sequences of
strongly-convex functions.

For $\alpha$-tame sequences of games bi-linear zero-sum games, Theorem~8 in
\citet{zhang2022no-regret} implies a tracking bound of order $T^{(3 + \alpha) / 4}$.
Our results do not apply in this setting in general, since there are no known
contractive algorithms for that class of problems. However, as discussed in
Example~\ref{example:EGD}, the extra-gradient method is contractive when initialized in
a small enough neighborhood of the equilibrium of the game. Hence our results suggest a
possible improvement on their bounds for some benign sequences of games, in which the
equilibria stay within the basin of attraction of previous equilibria. We leave this
investigation for future work.
Note also that the algorithm of \citet{zhang2022no-regret} is built on a meta-learning
framework, and maintains simultaneously $\log T$ algorithms over $T$ rounds, which may
come with an unwieldy memory and computation cost.

  \citet{anagnostides2024convergence}
  consider time-varying games and minimize the cumulative squared differences between the
  players' choices and their best responses. Under the additional assumption that each time-varying game is played for multiple rounds, they obtain instance-dependent bounds with
  respect to the first-order variations of the game plus the variation of the $\epsilon$-Nash
  equilibria (cf. Theorem 3.5). However, the assumptions restricting the time-variance of the game in this work are not comparable to our setting in Theorem 1.

Finally, tracking guarantees are also studied in the context of time-varying functions
or games with stochastic feedback. Note that this setting is fundamentally more
challenging and the algorithms and analysis techniques applied in this work are
tailored to these specific challenges.
In the case of time-varying strongly monotone games, sublinear tracking bounds of
order~$T^{(2+\alpha) / 3}$ are given in Theorem~2 and Corollary~3 in
\citet{Duvocelle_2023}.
Another example for optimization with stochastic feedback is
\citet{besbes2015non-stat}: given a tame time-varying strongly convex optimization
problem, their Theorem~4 bounds the dynamic regret $\sum_{t=1}^T f_t(x_t) -
f_t(x^\star_t)$ by a term of order $T^{(1+\alpha)/2}$. This directly implies sublinear
tracking guarantees. Again, we leave the study of stochastic feedback for future work.  
  The following table compares the most relevant existing results.

\begin{tabular}{|p{3.5cm}||p{1.5cm}|p{10cm}|}
 \hline
 & Rate & Setting\\
 \hline
\vspace{0.5em} \citet{Duvocelle_2023} & \vspace{0.5em}$T^{\frac{2+\alpha}{3}}$ &
 \vspace{0.3em}
   Stochastic feedback \newline
   Strongly monotone time-varying games.
 \vspace{0.7em}
 \\
  \hline
 \vspace{0.5em}  \citet{besbes2015non-stat}& \vspace{0.5em} $T^{\frac{1+\alpha}{2}}$ &
 \vspace{0.3em}
    Stochastic feedback\newline
     Strongly convex time-varying functions.
      \vspace{0.7em}
      \\
 \hline
 \vspace{0.5em} \citet{zhang2022no-regret}& \vspace{0.5em} $T^{\frac{3+\alpha}{4}}$ &
 \vspace{0.3em}
  Deterministic feedback\newline
     Time-varying bi-linear games (not strongly monotone).
   \vspace{0.7em}
  \\
 \hline
   \vspace{0.5em} \citet{Mokhtari2016Online-o}&  \vspace{0.5em} $T^{\alpha}$  &
  \vspace{0.3em}
    Deterministic feedback\newline
     Time-varying strongly convex functions.
    \vspace{0.7em}
 \\
 \hline
  \vspace{0.5em}  Theorem \ref{thm:contraction2track}&   \vspace{0.5em} $T^{\alpha}$  &
  \vspace{0.5em}
    Deterministic feedback \newline
     Time-varying VIs - any contractive algorithm.
 \vspace{0.7em}
   \\
 \hline
\end{tabular}

\section{Periodic time-varying variational inequalities and tracking.}\label{sec:periodic}
All tracking guarantees of this section are under the assumption that the time-varying
variational inequalities have a periodic solution (see Definition~\ref{def:periodicVIP}), and that 
the operators are all $\mu$-strongly monotone. Additional assumptions on $\cZ$ or $\cF$ are specified at the beginning of each subsection. 

\paragraph{Vacuousness of tracking guarantees via contraction.} 
We start by showing that the results from Section~\ref{sec:contraction}, although
applicable, are not sufficient to guarantee tracking for periodic
time-varying VI problem. Indeed, consider a $k$-periodic time-varying VI problem. Applying
Theorem~\ref{thm:contraction2track}, the bound we obtain on $\track_T(\cA)$ for 
$C$-contractive algorithms is
\begin{align}
\label{ineq:boundContractionPeriodic}
  \frac{1}{(1 -  C)^2} \sum_{t=1}^{T-1} \|Z_t^\star - Z_{t+1}^\star\|^2 
  + \frac{\|Z_1 - Z_1^\star\|^2}{1 -  C} 
  \geq \frac{1}{(1-C)^2} \floor{\frac{T}{k}} \sum_{i=1}^k  \|Z_i^\star - Z_{i+1}^\star\|^2 
  + \frac{ \|Z_1 - Z_1^\star\|^2}{1-C} .
\end{align}  
This bound is linear in $T$ as soon $Z_i^\star \neq Z_{i+1}^\star$ for some $i\in[k]$.
However, since the periodic problem is inherently easier than a non-periodic time-varying
VI problem, we would expect better tracking guarantees.

Our goal in this section is to improve on the vacuous bound in
\eqref{ineq:boundContractionPeriodic} for periodic time-varying variational inequalities.
We propose new algorithms, based on the aggregation of forward{-backward} methods, that
leverage the periodicity of the data. The results of this section are divided into two
parts:
\begin{enumerate}
  \item  In Section \ref{sec:metaPeriodic}, we provide a worst-case regret bound for
        bounded domains $\cZ$, which hold for any sequence of strongly monotone operators
        $F_t$, against any periodic sequence of comparators (cf.
        Theorem~\ref{thm:aggregation}). We apply this result in
        Corollary~\ref{cor:tracking_periodic} to derive tracking bounds when the
        comparators are specified to be the solutions of an exact periodic sequence of
        strongly monotone VI problem.
  \item Complementing these results, we show in Section~\ref{sec:unconstrained_lip} that for $\mathrm{VIP}((F_t),\cZ)$ with periodic solutions constant tracking bounds are achievable under the additional assumption that the $F_t$'s are Lipschitz continuous (cf.\,Theorem~\ref{thm:constRegretAdaHedge}). These results hold even for unbounded domains $\cZ$. 
    \end{enumerate}

\subsection{Aggregation of online gradient descent.}
\label{sec:metaPeriodic}
We now present our first set of results for tracking solutions of periodic  VI problems.
At the same time, we introduce the algorithmic framework that will also be used in
Section~\ref{sec:unconstrained_lip}. 
In this section, we assume that the domain $\cZ$ is bounded, with diameter
$D = \max_{Z_1, Z_2 \in \cZ}\|Z_1- Z_2\|$, and that the sequence of operators is bounded
by a known constant $G \geq \max_{Z \in \cZ}\|F_t(Z)\|$. 
 
\paragraph{On (dynamic) regret bounds for variational inequalities.}
The guarantees we derive in this section are upper bounds on quantities of the form
\begin{equation}
  \label{eq:regret-explanation}
  \sum_{t=1}^{T} \langle F_t(Z_t), Z_t - C_t\rangle - \frac{\mu}{2} \| Z_t - C_t\|^2 \,, 
\end{equation}
where $\mu \geq 0$ is a positive number, typically the strong-monotonicity constant 
of the operators. 
In analogy with the online optimization terminology, we call those guarantees
\emph{dynamic regret bounds} for the sequence of operators $(F_t)$ against the sequence
of comparators $(C_t)$. Note that, unless specified otherwise, we make no assumptions on
the way the operators $F_t$ are generated: they could be chosen by an adversary. The
periodicity of the data is captured by the assumption that the comparators are periodic.

As explained in the proof of Corollary~\ref{cor:tracking_periodic}, when the comparator
$C_t$ at time $t$ is specified to be the $k$-periodic solutions to the operators
$(F_t)$, and when the $F_t$ are $\mu$-strongly monotone, then the dynamic regret is an
upper bound on the tracking error.
 
\vspace{0.5em}
\begin{remark}
Upper bounding the dynamic regret \eqref{eq:regret-explanation} can be useful 
beyond the application we give.
For example, in the two fundamental examples of convex optimization and convex-concave
saddle-point optimization, under an additional strong monotonicity assumption, the
dynamic regret~\eqref{eq:regret-explanation} is an upper bound on the natural measure
of performance, as explained in Examples~\ref{ex:optim} and \ref{ex:gap_saddle_point}.
\end{remark}
\begin{remark}
Reducing the computation of solutions of VI problems to the optimization of scalar
objectives is standard in the VI literature. In fact, the summand in
\eqref{eq:regret-explanation} is a special case of a saddle-function associated to the
operator $F_t$, as discussed extensively in \citet{larsson1994class}; precisely, the saddle-function 
$
\mathcal L : (z, c) \mapsto 
  \langle F(z), z - c\rangle - (\mu / 2) \| z - c\|^2 \,.
$
This saddle-function naturally defines a `gap function', $\sup_{c \in \cX} \mathcal
L(z, c)$, see also \citet{Facchinei_2007} (Definition 10.2.2), which can serve as an
optimization proxy for the VI.
\end{remark}

\vspace{1em}

 \begin{example}[Optimization]
  \label{ex:optim}
  If $f$ is a $\mu$-strongly-convex function, with gradient $F = \nabla f$, then by strong
  convexity, for any $z, c \in \cZ$,
  \[
    f(z) - f(c)
    \leq \langle F(z), z - c\rangle - \frac{\mu}{2} \|z -c \|^2.
  \]
\end{example}
\begin{example}[Saddle-point Optimization]
  \label{ex:gap_saddle_point}
  If $u : (p_1, p_2) \mapsto u(p_1, p_2)$ is  a $\mu$-strongly-convex-strongly-concave function,
  with pseudo-gradient $F = (\nabla_1 u, -\nabla_2 u)$,
  then for any $z = (z_1, z_2)$ and $ c = (c_1, c_2) \in \cZ$, we have:
  \begin{align*}
    u(z_1, c_2)  -  u(c_1, z_2)     
    &=  u(z_1, c_2)  - u(z_1, z_2) + u(z_1, z_2) -  u(c_1, z_2)  \\
    & \leq -\langle \nabla_2 u (z_1, z_2),  z_2 - c_2 \rangle - \frac{\mu}{2} \| z_2 - c_2 \|^2
    + \langle \nabla_1 u (z_1, z_2), z_1 - c_1 \rangle  - \frac{\mu}{2} \| z_1 - c_1 \|^2 \\
    & = \langle F(z), z - c\rangle - \frac{\mu}{2} \|z - c\|^2.
  \end{align*}
\end{example}
In these two central examples, regret guarantees imply upper bounds on the cumulative gaps
against the sequence of comparators $(C_t)_{t \in \NN}$. Note also that this notion has
previously appeared in the online learning literature, e.g.\,in
\citet{zhang2022no-regret}.

\paragraph{General structure of aggregation and base algorithm.}
We use a \emph{meta-algorithm template}, as stated in Algorithm~\ref{alg:meta}. The basic idea is
to maintain several copies of a \emph{base algorithm} with different parameters.
The outputs of the base algorithms are combined using an aggregation algorithm,
guaranteeing performance comparable to that of the best base algorithm.  We refer to this aggregation algorithm as the \emph{meta algorithm}.
{In our setting, the base algorithms are initialized with different period lengths $i \in [K]$, i.e., $\cA_{\mathrm{cFB}}^{(1)}, \dots, \cA_{\mathrm{cFB}}^{(K)}$, and the meta algorithm guarantees a performance comparable to an algorithm initialized with the correct cycle length $k$, $\cA_{\mathrm{cFB}}^{(k)}$. 
This idea is commonly used in online learning for adapting to unknown parameters in
step-size tuning, e.g., \citet{van-erven2021metagrad}.
The unknown parameter in our case is the period $k$ of the sequence. }

\paragraph{Base algorithm.}
The base algorithm we consider is an instance of gradient descent tailored to
periodic problems, instead of stationary problems. Given a period $i$, the algorithm
maintains $i$ independent iterates and updates them cyclically. We denote by $\tilde F_t$
the operator given as input to the algorithm: this is to emphasize the fact that we will
sometimes give a surrogate operator to the algorithm instead of the true operator $F_t$
(see~\eqref{eq:surrogate}); we also instantiate it with the true operator in
Section~\ref{sec:unconstrained_lip}.
 
\SetKwInput{kwInit}{Initialization}
\SetKwInput{kwInput}{Input}
\begin{algorithm}
\caption{Algorithm $\cA_{\rm cFB}^{(i)}$: Cyclic Forward{-Backward} Method}\label{alg:cyc-OGD}
\kwInput{Assumed period $i$, step-size schedule $(\eta_s)$, { flag if used as a base algorithm.} }
\kwInit{Initial iterates $(Z_{1, 1}, Z_{2,1}, \dots, Z_{i, 1})= (Z_1, Z_1, \dots, Z_1)$ with $Z_1 \in \cZ$.}
\For{$t =1, \dots, T$}{
Pick current index $n = (t \bmod i) + 1 $ , and corresponding time $s = \lceil t / i \rceil$ \;
Play $Z_t = Z_{n,s} $, receive operator $\tilde F_t$ \;
{\If{ $\cA_{\rm cFB}^{(i)}$ functions as a base algorithm }{Forward $Z_t$ to meta-algorithm\;}}
Update $Z_{n, s+1} = \Proj_{\cZ}(Z_{n, s} - \eta_s \tilde F_t(Z_{n, s}))$\;
}
{\textbf{Output: } sequence $(Z_t)$\;}
\end{algorithm}
 
We first recall the analysis for a fixed comparator. This is a direct application of the
online gradient descent bound on the sequence of linear losses $z \mapsto \langle
\tilde F_t(Z_t), z \rangle $, used in the analysis of online gradient descent
for strongly-convex losses; 
see Theorem 2.1 in \citet{bartlett2007adaptive}.

\begin{proposition}
  Given $\mu > 0$, for any fixed $C \in \cZ$, if $\eta_t = 1 / (\mu t)$, then
  the Online Forward{-Backward} Method (Algorithm~\ref{alg:cyc-OGD} tuned with period $i=1$), ensures that
for all $t\in [T]$, 
   \[
    \sum_{t=1}^{T} \langle \tilde F_t(Z_t), Z_t - C\rangle - \frac{\mu}{2} \| Z_t - C\|^2
    \leq \frac{(\max_{t\in [T]} \|\tilde F_t(Z_t)\|)^2}{2\mu} \big( \log (T) + 1\big) \,.
  \]
\end{proposition}
A direct consequence of this bound is that if we tune the learner with a period $k$,
then for any $k$-periodic sequence of comparators $(C_t)_{t \in { \NN}}$ {with $C_{t} = C_{t+k}$ for all $t \in \NN$}, the regret of cyclic gradient
descent can be upper bounded by separating the update times of each individually
maintained iterate.
\begin{corollary}\label{cor:periodic_gd}  Let $(F_t)_{t \in \NN}$ be a sequence of $G$-bounded operators. 
For any $k$-periodic sequence of comparators  $(C_t)_{t \in { \NN}}$ {with $C_{t} = C_{t+k}$ for all $t \in \NN$}, Algorithm~\ref{alg:cyc-OGD} tuned
with the period $k$ and step-size schedule $\eta_s = 1 / (\mu s)$ played on $\tilde F_t =
F_t$ enjoys the bound 
  \[
    \sum_{t=1}^{T} \langle F_t(Z_t), Z_t - C_t\rangle - \frac{\mu}{2} \| Z_t - C_t\|^2
    \leq \frac{kG^2}{2\mu} \big(\log (T/k)  + 1\big) \,.
  \]
\end{corollary}
We note that this result holds in particular if the sequence of
comparators is the ($k$-periodic) sequence of solutions, that is $C_t=Z_t^\star$.
 Since $(F_t)$ are by assumption $\mu$-strongly monotone, this implies that the tracking
error is upper bounded by the regret. Thus, $\track_T(\cA_{\rm{cFB}}^{(k)})$ is in the
order of $ (G^2/\mu^2) k\log T$. 
In practice, we expect the period of the data to be unknown. Therefore, we run several
instances of the base algorithm. Each instance is initialized with a different period
length, and the iterates of the base algorithms are aggregated by a meta-algorithm.
\paragraph{Aggregation.}
We maintain $K\geq k$ base algorithms and aggregate them using an expert algorithm.We wish
to retain a regret bound for the combined aggregated iterates that is comparable to that
of the best algorithm in hindsight. The standard Hedge bound is of order $\sqrt{T\ln K}$,
which would give an overhead cost that dominates the bound of the best-performing
algorithm. 
In order to retain the fast logarithmic rate, we propose a carefully crafted aggregation
scheme working for VI problems, based on a combination of surrogate operators with exponentially
weighted aggregation (a.k.a. Vovk's Aggregating Algorithm in this context, see
\citet{vovk1995a-game-o}). The $K$ base algorithms are updated with the affine surrogate
operators
\begin{equation}\label{eq:surrogate}
  \tilde F_t : z \mapsto F_t(Z_t) + \mu (z - Z_t) \,.
\end{equation}
 The term `surrogate operators' is inspired by the online learning terminology. 
From each base algorithm  $\cA_{\rm cFB}^{(i)}$, the meta-algorithm receives an iterate $Z_t^{(i)}$. 
Then, the meta-algorithm combines the base predictions $Z^{(i)}_t$  by maintaining a
probability distribution over the $K$ base algorithms. This probability distribution is given by the aggregation algorithm. We define for all $i\in[K]$, 
 \begin{equation}\label{eq:exp_weights}
  p_{t, i} \propto \exp\biggl(-\lambda \sum_{s=1}^{t-1} \ell_{s, i} \biggr)
  \quad \text{with} \quad
  \ell_{s, i} = \langle F_s(Z_s), Z_s^{(i)} \rangle + (\mu / 2)\|Z_s^{(i)} - Z_s\|^2 \,.
\end{equation}
The learning rate is set to $\lambda = (4\mu(D + G/\mu)^2)^{-1}$. This learning
rate depends on the diameter of the set $\cZ$, i.e., $D = \max_{Z,Z' \in \cZ}
\norm{Z-Z'}$, and $G \geq \max_{t \in T}\norm{F_t(Z_t)}$
 and the strong-monotonicity parameter $\mu$. We then select the final
action $Z_t = p_{t, 1} Z_t^{(1)} + \dots + p_{t, K} Z_t^{(K)}$.

\begin{algorithm}[h!]
\caption{ Meta-Algorithm Template }\label{alg:meta}
\KwData{
Maximum cycle period $K \in \NN$, 
aggregator over $K$ base iterates, 
$K$ base algorithms. 
}
 \For{t = 1, \dots, T}{
 Receive base plays $(Z^{(1)}_t, \dots, Z^{(K)}_t)$ from base algorithms {(Algorithm \ref{alg:cyc-OGD})} \;
 Receive distribution $p_t$ over $[K]$ base algorithms from aggregator algorithm \;
 Play $Z_t = p_{t, 1} Z^{(1)}_t + \dots +p_{t, K} Z^{(K)}_t$ and receive $F_t$ \; 
 Update aggregator algorithm with losses $\ell_{i, t}$ \;
 Update base algorithms {(Algorithm \ref{alg:cyc-OGD})} with operator $\tilde F_t$ \; 
 }
\end{algorithm}

\begin{theorem}
\label{thm:aggregation} 
 Let $(F_t)_{t \in \NN}$ be a sequence of $G$-bounded operators
on the bounded domain $\cZ$ of diameter $D$.
Consider Algorithm~\ref{alg:meta}, {initialized} with maximum period $K$, 
using exponential weights (with losses specified in \eqref{eq:exp_weights}, 
and learning rate $\lambda = (4\mu(D + G/\mu)^2)^{-1}$)  
and cyclic forward{-backward} method (Algorithm \ref{alg:cyc-OGD}) with step-size schedule $\eta_s =
1 / (\mu s)$ as base algorithms. 
For any $k \leq K$, against any $k$-periodic sequence of comparators $(C_t)_{t \in { \NN}}$ {with $C_{t} = C_{t+k}$ for all $t \in \NN$} 
  \[
    \sum_{t=1}^{T} \langle F_t(Z_t), Z_t - C_t\rangle - \frac{\mu}{2} \| Z_t - C_t\|^2
    \leq \frac{(G + \mu D)^2}{2\mu} \Big( k \log (T/k) + k
    + 8 \log K \Big) \,.
  \]
\end{theorem}
Interestingly, since the mappings $z \mapsto \langle F_t(z), z\rangle$ are not convex, our
proof requires careful manipulations of the surrogates. A positive computational
side-effect of using these surrogates is that the full procedure only requires a single
evaluation of the operator $F_t$. For a complete proof, see
Appendix~\ref{sec:appendix:proofThm3}.
\begin{corollary}[Tracking periodic variational inequalities]
\label{cor:tracking_periodic}
Let $(F_t)_{t \in \NN}$ be a sequence of $\mu$-strongly monotone and $G$-bounded operators
on the bounded domain $\cZ$ of diameter $D$. If the solutions of $(F_t)$ are $k$-periodic, then
Algorithm~\ref{alg:meta} tuned as in the statement of Theorem~\ref{thm:aggregation}
guarantees
\[
  \sum_{t=1}^{T} \|Z_t - Z_t^\star\|^2
  \leq \frac{(G + \mu D)^2}{\mu^2} \Big( k \log (T/k) + k
  + 8 \log K \Big) \,.
\]
\end{corollary}
\begin{proof}{Proof.} 
Note that $\dprod{F(Z_t^\star),Z_t-Z_t^\star}\geq 0$ by definition.  
With $\mu$-strong monotonicity, we have
\[
  \mu\| Z_t - Z_t^\star\|^2
  \leq \langle F_t(Z_t) - F(Z_t^\star), Z_t - Z_t^\star \rangle
\leq \langle F_t(Z_t), Z_t - Z_t^\star \rangle \,, 
\]
which we rearrange to make the instantaneous regret appear
\[
\frac{\mu}{2} \| Z_t - Z_t^\star\|^2 
\leq 
\langle F_t(Z_t), Z_t - Z_t^\star \rangle
- \frac{\mu}{2} \| Z_t - Z_t^\star\|^2  \,.
\]
Therefore, after summing over $t\in [T]$, Theorem~\ref{thm:aggregation} provides the
claimed result.
\qed
\end{proof}

\subsection{Constant tracking error under Lipschitzness.}
\label{sec:unconstrained_lip}
The previous results are limited to the constrained setting, in which the domain $\cZ$
is bounded. They also require the knowledge of an upper bound $G$ on the norm of the
point values of the operators. We now consider the unconstrained setting, $\cZ = \R^d$,
and add a Lipschitz assumption on the operators $F_t$. Note that we still assume that the
operators are all $\mu$-strongly monotone. Furthermore, in contrast to
Section~\ref{sec:metaPeriodic}, we allow for $K$ evaluations of the function $F_t$. Our
algorithm is another instance of the meta-algorithm framework introduced in
Section~\ref{sec:metaPeriodic}. In this Section, we give an algorithm that obtains a constant
tracking bound that is independent of $T$.

\paragraph{Base algorithms.} 
To take advantage of the Lipschitz assumption, we use variants of the cyclic forward
algorithm with a constant learning rate of $\eta= 1/L$ and feedback $\tilde F_t= F_t$.
Note that using this feedback now requires $K$ evaluations of $F_t$. However, the
constant learning rate together with full multi-point feedback gives exponentially fast
convergence to the solution.
\begin{lemma}
  Let $\mathrm{VIP}((F_t), \cZ)$ be a time-varying VI problem with $k$-periodic solutions. Assume
  the functions $F_t$ are $\mu$-strongly monotone and $L$-Lipschitz.
  Let $(Z^{(k)}_t)$ denote the sequence of iterates of the forward
  algorithm $\cA_{\mathrm{cFB}}^{(k)}$ (Algorithm~\ref{alg:cyc-OGD}), tuned with $\eta_t =
    1/L$ and correct period $k$. Then
  \[
      \big\| Z_t^{(k)} - Z^\star_{t} \big\|^2
    \leq \Big(1  - \frac{\mu}{L} \Big)^{\lfloor t/k \rfloor} 
      \big\| Z_1 - Z^\star_{t} \big\|^2 \,.
  \]
\end{lemma}
The result is a small modification of Example~\ref{example:GD} in Section \ref{sec:examplesContraction}. 
 From this, we obtain that the base
algorithm tuned with the correct period $k$ converges exponentially fast
to the solutions $Z_1^\star, \dots, Z_k^\star$. 

\paragraph{Aggregation with adaptive step-sizes.}
 The distribution over the base algorithms is an instance of exponential weights, with the
 crucial difference with respect to our approach of Section~\ref{sec:metaPeriodic} that we
 use an adaptive learning rate $\lambda_t$. More precisely, let $p_1$ be the uniform
 distribution over $[K]$. Define the loss vector and the loss of the average play at time
 $t$, 
\begin{equation}
  \label{eq:ew_and_losses}
  \ell_{t, i} = \langle F_t(Z_t), Z_t^{(i)} \rangle + (\mu / 2)\|Z_t^{(i)} - Z_t\|^2  
  \quad
  \text{and}
  \quad
  \bar \ell_t =  \langle F_t(Z_t), Z_t \rangle
  \,.
\end{equation}
For the first rounds, if $\bar \ell_t \leq \min_{i} \ell_{t, i}$, define $p_{t+1}$ to be
the uniform distribution over the minimal components of $\ell_1 + \cdots + \ell_t$. Let
$T_0$ be the first time at which $m_t < \bar \ell_t$. For $t \geq T_0 + 1$, denoting by
$(x)_+ = \max (x, 0)$ the positive part of $x$, set $p_t$ to be the distribution over
$[K]$ such that
\begin{equation}
  \label{eq:ada_lr}
  p_{t, i} \propto \exp\biggl(-\lambda_t \sum_{s=1}^{t-1} \ell_{s, i} \biggr)
  \quad 
  \text{where}
  \quad
  \lambda_t = \frac{\log K}{ \sum_{s=T_0}^{t-1} (\bar \ell_s - m_s)_+} 
  \quad
  \text{and}
  \quad
  m_s = -\frac{1}{\lambda_s} \sum_{i = 1}^K p_{t, i} e^{-\lambda_s \ell_{s, i}} \,.
\end{equation}
Note that the denominator in the definition of the learning rate is positive, by
definition of $T_0$. 
The initialization steps, when $t \leq T_0$ can be informally described as setting
$\lambda_t = \infty$ when that denominator is $0$.
The learning rate is tuned in the style of AdaHedge (\citet{De-Rooij:2014aa}), to decrease
proportionally to the inverse of the cumulative gap between the
loss of the average of plays $\bar \ell_s$, and the mix loss~$m_s$.

In the unconstrained setting, the final convergence guarantees unavoidably depend on the initial point chosen by the algorithm. Given a time-varying problem
$\mathrm{VIP}((F_t),\cZ)$ with periodic solutions, let $\eq$ denote the set of all the
solutions of the operators $(F_t)$ (which has cardinality at most $k$). For simplicity, we
assume that all base algorithms are initialized with the same $Z_1 \in \cZ$, and define
the initial distance
\[
  D_0 = \max_{Z^\star \in \eq} \|Z_1 - Z^\star\| \,.
\]
The next theorem provides constant tracking guarantees for our algorithm in the
unconstrained setting, depending on the initial distance and condition number $\kappa =
L / \mu $.
\begin{theorem}\label{thm:constRegretAdaHedge}
Consider a time-varying unconstrained problem $\mathrm{VIP}((F_t), \R^d)$ with
$k$-periodic solutions. Assume all operators $(F_t)$ are all $\mu$-strongly monotone and
$L$-Lipschitz, and let $\kappa = L / \mu$ denote the condition number of the problem. 
Algorithm~\ref{alg:meta}, tuned with $K\geq k$ and losses and learning rate specified in
\eqref{eq:ew_and_losses} and \eqref{eq:ada_lr}, has tracking guarantees:
   \[
    \sum_{t=1}^{T }\norm{Z_t - Z^\star_t}^2
\leq  4D_0^2 (2 + \kappa) \Big(2(2\kappa^2 + 1)(2 + \kappa)\log K  
+ (2\kappa + 1) \kappa k + 1\Big)\, . 
  \]

\end{theorem}
\smallskip

The proof is detailed in Appendix~\ref{app:proof_of_ada}. The crux of the proof, in
Lemmas~\ref{lem:agg_ada} and~\ref{lem:mix_gap_ada}, 
consists in showing that thanks to the adaptive learning rate, the aggregation method
ensures similar regret bounds as in the proof of Theorem~\ref{thm:aggregation}, but
replacing $D$ and $G$ by their empirical counterparts $\what D_T$ and $\what G_T$, defined as
follows:
 \[
\what D_T = \max_{{t \in [T], \,  i, j \in [K]}} \norm{Z_t^{(i)} - Z_t^{(j)}}
\quad
\text{and}
\quad
\what G_T = \max_{t \in [T], \,  i \in [K]} \norm{F_t(Z_t^{(i)} )} \,; 
\]
we show in Lemma~\ref{lem:bounded_d_and_g} in Appendix~\ref{app:proof_of_ada} that these quantities are controlled by $D_0$
for our algorithms and assumptions. The rest of the proof is similar to that of
Theorem~\ref{thm:aggregation}.
 
\section{Application Examples}\label{sec:examples}
In this section, we provide examples of time-varying VI problems to illustrate the practical relevance of our results. 
As mentioned in the introduction, it is well known that VIs are a generalization of optimization and saddle-point problems. Examples \ref{Example:KellyAuction} and \ref{ex:streaming_lin_reg} in Section \ref{section:exampleOptAndGames} are illustrations of practical problems in these settings. See also \citet[Chapter I]{facchinei2003finite-d} for many more examples, which often have natural time-varying extensions. 
  {In Section \ref{section:GLM}, we give an example illustrating that VIs can also be applied to statistical parameter estimation.}
{In its full generality, this problem is not equivalent to an optimization problem, thus showcasing the wider scope of VIs.}
\subsection{Examples: Time-varying convex optimization and saddle-point problems.}\label{section:exampleOptAndGames}
 For clarification and completeness, we recall the relation between VIs and convex optimization and saddle-point problems. Let $\cZ$ be a closed convex set and let $f:\cZ \rightarrow \RR$  be a differentiable and convex function.  In the following examples, we will either assume that $\cZ$ is bounded or that $f$ is strongly convex. Under either of these assumptions, 
if we set $F = \nabla f$ then $u^{\star}$ is a solution of $\mathrm{VIP}(F,\cZ)$ if and
only if
\begin{align}
\label{OPT}
u^\star \in \argmin_{x \in \cZ } f(x). \tag{OPT} 
\end{align}
If the function $f$ is $\mu$-strongly convex and $L$-smooth, then the corresponding VI problem
is $\mu$-strongly monotone and $L$-Lipschitz.
  To illustrate
the relation to game theory, consider an $n$-player normal-form game with convex loss
functions: Let $\cZ^{(i)}$ denote the (convex) action set of player $i\in[n]$, and define 
$\cZ = \prod_{i \in [n]} \cZ^{(i)}$, and
$\cZ^{(-i)} = \prod_{i\in[n]\setminus \{i\}}\cZ^{(i)}$.
Finally, denote by $\nu^{(i)}:\cZ^{(i)} \times \cZ^{(-i)} \rightarrow \RR$ the $i$-th player's loss function and assume that for all $x^{(-i)} \in \cZ^{(-i)}$,
the loss $\nu^{(i)}(\, \cdot\, , x^{(-i)})$ is convex and differentiable. A point $\bar x
= [\bar x^{(1)}, \dots , \bar x^{(n)}] \in \cZ$ is a \emph{Nash equilibrium} if
\begin{align}
  \label{NE}
  \forall i \in [n], \; \forall x^{(i)} \in \cZ^{(i)}: \quad \nu^{(i)}\left(x^{(i)},\bar x^{(-i)}\right) \geq \nu^{(i)}\left(\bar x^{(i)},\bar x^{(-i)}\right). \tag{NE}
\end{align}
Consider the operator $F$ mapping $(x^{(1)}, \dots , x^{(n)})$ to 
$[\nabla_{x^{(1)}}\nu^{(1)}(x^{(1)}, x^{(-1)}), \dots,
\nabla_{x^{(n)}}\nu^{(n)}(x^{(n)}, x^{(-n)})]$, also called the \emph{pseudo-gradient} 
of the game $(\nu^{(i)})$. 
The set of Nash equilibria is exactly the set of solutions of $\mathrm{VIP}(F,\cZ)$. 
See, e.g., \citet[Chapter~12]{Rockafellar:1998aa},   or \citet[Chapter 7]{Kinderlehrer:2000aa},   for more details and examples.

\begin{example}[Repeated kelly auctions on seasonal markets] 
\label{Example:KellyAuction}We adapt an example from \citet{Duvocelle_2023}. Consider a repeated Kelly auction (\citet{Kelly:1998aa}) where $n$ bidders participate in a repeated auction on a limited good. Let $R_t \geq 0$
denote the minimal price, i.e., the entry barrier set at time $t$. 
For all $i \in [n]$, given a maximum budget $b^{(i)}_{t}$, and their value for the good,
$v_t^{(i)}\geq 0$, player $i$ bids $x^{(i)}_{t} \in [0,b^{(i)}_{t}]$. 
Each player gets a fraction of the good relative to their bid. Their utility function 
is 
\begin{align*}
\nu^{(i)}_t\big(x_t^{(i)}, x_t^{(-i)}\big) 
= v_t^{(i)}\frac{x_t^{(i)}}{R_t + \sum_{j\in[n]}x_t^{(j)}} -x_t^{(i)} \,.
\end{align*}
For many goods, the demand or the supply depends on exogenous factors that vary with 
time. For example, the availability of agricultural products depends on the season, and 
the demand for luxury goods may depend on fashion trends. These conditions naturally affect
the minimal price $R_t$, and the individual values $v_t^{(i)}$, making the game 
time-varying with seasonal shifts. For further examples and context, also see Section 2 in \citet{rahme2021auction}.
\end{example}
\smallskip
   To apply the results from Section \ref{sec:contraction}, assume that the demand adheres to a slowly shifting fashion trend, thus satisfying tameness (cf. Definition \ref{def:tameVIP}). Furthermore, note that the pseudo-gradient is Lipschitz and monotone; thus, adding a quadratic regularizer to the utility guarantees strong monotonicity. Hence, combining Example \ref{example:GD} and Theorem \ref{thm:contraction2track} provides tracking guarantees for this example. 
  Further, consider a scenario where the availability of a good follows an idealized seasonal trend, hence satisfying periodicity. Again, consider the regularized utilities. Given that we have an upper estimation on the periodicity, Theorem \ref{thm:constRegretAdaHedge} provides a constant tracking bound. We note that the assumption of periodicity is rather unrealistic for many real-world applications. Hence, this example also illustrates the practical relevance of the future research directions in Section \ref{sec:unconstrained_lip}.

 \begin{example}[Linear regression on streaming data] \label{ex:streaming_lin_reg}
Consider a linear regression problem with a continuous stream of i.i.d.\ data as
introduced in \citet{Foster:1991aa}. That is, let $ \{(a_1,b_1), \dots,
  (a_{n_t},b_{n_t})\}\subset \RR^{ d} \times \RR$ denote the data at time $t$, and $A_t :=
  [a_1,\dots a_{n_t}]^\top \in \RR^{n_t \times d}$ and $b_t := [b_1,\dots,b_{n_t}]^\top\in
  \RR^{n_t}$ the corresponding data matrix and vector, respectively. Suppose the data
points $(a_i,b_i)$ are coming from a stream of measurements of the same experiment. In an
experimental setting, it can be beneficial to have a real-time regression on the data
stream (see e.g.\,the software tool~\citep{MatLab:2020aa}). Hence, define $f_t:\RR^d
  \rightarrow \RR$ with $f_t(x) = \norm{A_t x - b_t}^2 + \lambda \norm{x}^2$. Under the
assumption that the experimental data follows a fixed sub-gaussian distribution, the
solution to the empirical regression problem approaches the solution with respect to the
true distribution with high probability, see \citet[Section~13]{Shalev-Shwartz:2014aa}.
This defines an optimization problem with a slowly shifting solution and thus a
\emph{tame time-varying} optimization problem.
\end{example} 

\smallskip
 We note that the objective function of this tame time-varying optimization problem is $\lambda$-strongly convex; hence, combining Example \ref{example:Hilbert} and Theorem \ref{thm:contraction2track} provides tracking guarantees for this example. Note that Example \ref{example:Hilbert} does not require $\cZ$ to be bounded. Hence, it fits this example where $x \in \RR^d$ is from an unbounded set.

\subsection{Example: Generalized Linear Models}\label{section:GLM}
To illustrate the generality of VIs and their time-varying variant, we conclude this section with an example from statistical parameter estimation. 
\begin{example}[Generalized linear models]\label{Example:GLM} 
  We present an example from \citet{Juditsky:2019}. Consider a generalized linear model corrupted by Gaussian noise; i.e., we
  obtain observations $(a_t,b_t)\in \RR^d\times \RR$, $t=1,\ldots,T$, that satisfy $b_t=
  \phi( \langle Z^{\star},a_t\rangle)+ \xi_t$, where $\phi:\RR\mapsto\RR$ is a continuous
  and increasing \emph{link function}, $\xi_t\sim {\cal N}(0,\sigma^2)$, and
  $Z^{\star}\in \RR^d$ is an unknown parameter vector.
\citet{Juditsky:2019} observe that this parameter estimation problem can be reformulated as a monotone VI problem $\mathrm{VIP}(F,\RR^d)$ with operator  
\begin{equation} \label{eqn:monotone_oper_MLE}
F(Z) = \frac1T\sum_{t=1}^T a_t\, \left[\phi\left(\dprod{ Z,a_t}\right)- b_t\right] 
\,.
\end{equation}
Moreover, under minimal distributional assumptions, it can be proven that the
population-type operator (i.e., the expected value of $F$ under the target
distribution) becomes strongly monotone.
{
  For simplicity of presentation, we focus here on the univariate case where the link
  function $\phi$ maps from $\R \to \R$. In this specific setting, the operator $F$ in
  \eqref{eqn:monotone_oper_MLE} is always the gradient of a convex potential function,
  making the VI equivalent to a convex optimization problem. However,
  \citet{Juditsky:2019} consider a broader setting with multivariate link functions
  ($\phi : \R^m \to \R^m$), where the operator's Jacobian is not necessarily symmetric.
  Consequently, the general problem is not equivalent to an optimization problem and
  serves as a genuine example of the broader VI framework.
}
The above VI can also be considered in a streaming setting, such as proposed in Example \ref{ex:streaming_lin_reg}. If needed, regularization can be introduced by adding $\lambda Z$ to the operator in equation \eqref{eqn:monotone_oper_MLE}.
\end{example}
  Note that similarly to Example \ref{ex:streaming_lin_reg}, Theorem
  \ref{thm:contraction2track} combined with Example \ref{example:Hilbert} imply tracking
  guarantees.
\begin{remark} The above VI reformulation is one common option for solving GLM. Conversely, we may also consider the maximum likelihood estimator of the unknown parameter $Z\in \RR^d$ via solving the (nonlinear) least squares problem
\[ \min_{Z\in \RR^d}\Big\{ \frac1T\sum_{t=1}^T \left[b_t-\phi\left( \dprod{ Z,a_t}\right)\right]^2  \Big\}. \]
This optimization problem is typically nonconvex (with the exception of the case where the link function is affine). Thus, computing a solution might be computationally expensive. 
\end{remark} 
 

 \section{Main negative result: chaos.}\label{sec:chaos}
 
In this section, we provide a detailed case study of the convergence properties of fixed
step-size gradient descent on some periodic time-varying problems. 
We focus on optimization instead of the general VI framework since we aim to provide
simple examples of striking phenomena. For the same reason, we focus on periodic functions. Recall that this implies that the functions have a periodic solution (cf. Definition \ref{def:periodicVIP}). Accordingly, we always aim for the smallest
dimension in which these observations occur. 
\paragraph{Motivation and practical relevance}
Besides the surprising appearance of chaotic behavior, which we found interesting in its
own, we draw several practical conclusions from the results in this section.
First, proper parameter tuning is crucial: while fixed step sizes ensure exponential
convergence in strongly monotone and Lipschitz continuous variational inequalities,
ignoring periodicity in step size tuning can lead to instability. Second, in many
applications, optimization methods are used with step sizes exceeding the bounds for
which theoretical convergence results exist. Empirical studies (see
\citet{Cohen2021GradientDO}) show that larger step sizes can improve machine learning
model performance in both convergence and generalization.
 Our study adds a surprising piece to this puzzle: in periodic settings, increasing the
 step size can turn a divergent system into a convergent one. Given the popularity of
 single-shuffle stochastic gradient descent \citep{Mishchenko:2020aa}, which falls into
 the periodic setting, this finding might have practical implications.
{
However, we emphasize that the emergence of chaos, as discussed here, is a property of
the autonomous (time-invariant) dynamical system induced by the periodic composition of
operators. In general non-autonomous (arbitrary time-varying) or stochastic settings 
(using SGD), Li-Yorke chaos is not meaningful. 
Moreover, in practice, adaptive step-size methods (such as AdaGrad, AcceleGrad,
UnixGrad, UnderGrad, \citet{levy2018online, kavis2019unixgrad,
antonakopoulos2022undergrad} or Armijo-type line-search \citet{armijo1966minimization})
are often employed and can mitigate instability or chaotic behavior in deterministic
settings. Our results thus primarily highlight a theoretical phenomenon in
deterministic, periodic settings, and serve as a caution for using fixed step-sizes in
such cases. Investigating whether similar phenomena occur more generally in stochastic
systems, or with adaptive step sizes is left for future work.
}

 
\subsection{Notation}
Let $x_1 \in \cZ$ be an arbitrary point and consider the sequence of iterates generated
by gradient descent with step-size $\eta$ on a $k$-periodic sequence of $\mu$-strongly
convex and $L$-smooth problems $(f_t)_{t \in \NN}$, and define $F_t = \nabla f_t$. Define the
Gradient Descent iteration operator
\begin{align}
\label{def:GDoperator}
  \Phi_{\eta, t} : x \mapsto \Proj_{\cZ} \big(x - \eta F_t(x) \big)
  \quad \text{and} \quad
  \bar \Phi_{\eta} = \Phi_{\eta, 1} \circ \dots \circ \Phi_{\eta, k}
  \,.
\end{align}
With this notation, the sequence $(x_t)_{t\in \NN}$ satisfies the equations $x_{t+1}= \Phi_{\eta,
t}(x_t)$ for all $t\in\N$, and $x_{(n+1)k} = \bar \Phi_{\eta}(x_{nk})$ for all $k\in\N$.

Since all $\Phi_{\eta, t}$ are continuous (by continuity of Euclidean projections on
convex sets, together with the $L$-Lipschitzness of $F_t$), and since $\cZ$ is closed,
the sequence $(x_t)_{t\in\NN}$ is bounded if and only if one of its subsequences $(x_{i +
      k n})_n$ is bounded. Therefore, the boundedness of the iterates is solely dictated by $
  \bar \Phi_{\eta}$, and the starting $F_1$ does not matter for this discussion. We already provided an example of quadratic time-varying problems
with a constant minimizer---arguably the simplest setting. We observed
that the value of the step-size $\eta$ with the best speed of convergence may be larger
than values that lead to divergence. In the following section, we make a series of
striking observations, notably:
\begin{itemize}
  \item The sequence of GD iterates can converge to arbitrary points, go through cycles of
        arbitrary length, or even be chaotic, see Section~\ref{sec:chaosGD}.
  \item Period-halving bifurcations occur, interleaved with divergent phases, see
        Figure~\ref{fig:bifurcation}
  \item For specific choices of the step size, we observe star-shaped limit sets, see
        Section~\ref{sec:starAttractors}.
         
\end{itemize}

\subsection{Examples of chaotic behavior with large step-sizes.}
\label{sec:detailed_chaos}
We describe simple instances of time-varying problems which exhibit
exotic behavior of gradient descent with large step sizes.
We show that \emph{even in simple examples} the trajectory of Gradient Descent ---and its
dependence on the learning rate--- can be surprisingly complex. In particular, we show that chaos
can emerge in one-dimensional periodically time-varying optimization problems with a fixed minimizer. Moreover,
increasing the learning rate from a chaotic trajectory may lead to divergence, to cyclic 
behavior, or to converging trajectories. 

The examples we consider are built using the following instance of a time-varying
optimization problem. Given a periodic sequence of symmetric positive definite matrices
$A_t$, let $f_t$ be functions from $\R^d$ to $\R$ defined by
\begin{equation} \label{eq:exp}
  f_t(x) = \log\big(1 + e^{x^\top A_t x / 2} \big), 
  \quad
  \text{so that}
  \quad
  F_t(x) := f_t'(x) =   \frac{e^{x^\top A_t x / 2}}{1 + e^{x^\top A_t x  / 2}} A_t x \,.
\end{equation}
The functions $f_t$ are close to quadratic functions (and $F_T$ are close to
linear), but the non-linearity in the gradients is enough to produce chaotic
behavior. Note also that the minimizer is always $0$. Furthermore, note that $(F_t)$ is a sequence of $\mu$-strongly monotone and $L$-Lipschitz functions, where the parameters $\mu$ and $L$ depend on $(A_t)$.

\paragraph{Definitions: stability and chaos.}

Our description of the trajectory of gradient descent uses notions from real discrete
dynamical systems theory which we now recall. We refer the reader to the two monographs
\citet{Block2006Dynamics,Ruette2017Chaos-on} from which we extracted these definitions
and theorems.  Note that we do not state the results in full generality, for the sake of brevity. 
 
Following the notations introduced above, we denote a trajectory of gradient descent with
step-size $\eta$ by $(x_t)$ (the step-size will be clear from context). That is,
$x_{t+1}= x_t - \eta F_t(x_t)$, and $x_{k(n+1)} = \bar \Phi_\eta(x_{kn})$ and the points
at time steps that are multiples of the time-period $k$ are orbits of the stationary
dynamical system $(\bar \Phi_\eta, \R)$. (Note that $\bar \Phi_\eta$ is always
continuous.)

A point $x$ is said to be periodic under the map $\Phi$ if its (forward) orbit $(\Phi^n(x))_n$ is
periodic; the period of $x$ is the smallest period of the sequence. 
A periodic orbit $P = \{z_1, \dots, z_p\}$ of $\Phi$ is \emph{asymptotically stable} if
there exists an open set $U$ such that for any $x \in U$, the sequence of repeated
iterations $(\Phi^{(pn)}(x))_n$ converges to a point in~$P$. Then any trajectory
initialized in that open set will eventually get close to the orbit $P$. 

The results we are most interested in apply to cases in which the function $\Phi$ maps a
compact interval $I \subset \R$ into itself. Then, we shall say that the pair $(\Phi, I)$
forms an interval map. In that case, the orbits from a starting point inside $I$ stay
bounded in $I$.

Given an interval map $\Phi$ on $I$, a set $S \subset I$ is said to be \emph{scrambled}
if for any $x, y \in S$, we have $\limsup_{n\rightarrow \infty} | \Phi^n(x) - \Phi^n(y)| > 0$ and $\liminf_{n\rightarrow \infty} |
\Phi^n(x) - \Phi^n(y)| = 0$. The interval map $(\Phi, I)$ is said to be \emph{chaotic} in
the sense of Li-Yorke if there exists an uncountable scrambled set.

Li-Yorke chaos is one of the standard formalizations of chaos in real-valued dynamical
systems, and has been invoked recently in the game theory and optimization literature to
describe the trajectories of standard algorithms in games, e.g.\,in
\citet{pmlr-v139-bielawski21a}.

\subsection{Different behaviours of gradient descent.}\label{sec:chaosGD}
In our first example, we consider $d=1$ and a  periodic sequence of length $2$.  More precisely, in \eqref{eq:exp} let for all $n
\in \N$, 
\[
  A_{2n + 1} = 0.25
  \quad \text{and} \quad
  A_{2n} = 4 \,.
\]
Despite the very basic features of the problem, we show that depending on
the value of the learning rate, Gradient Descent can:
\begin{itemize}
  \item[-] converge to zero; the global minimizer of the sequence,
  \item[-] oscillate around cycles of points which are not minimizers, 
  \item[-] be chaotic in the sense of Li-Yorke,
  \item[-] diverge to infinity.
\end{itemize}
Surprisingly, increasing the learning rate can get the trajectory out of unstable
dynamics and recover convergence. This is summarized in the bifurcation diagram
of Figure~\ref{fig:bifurcation}, and in Observation~\ref{obs:gd_behaviours} below.

\begin{figure}[h]
  \center
	\includegraphics[width=.99\linewidth]{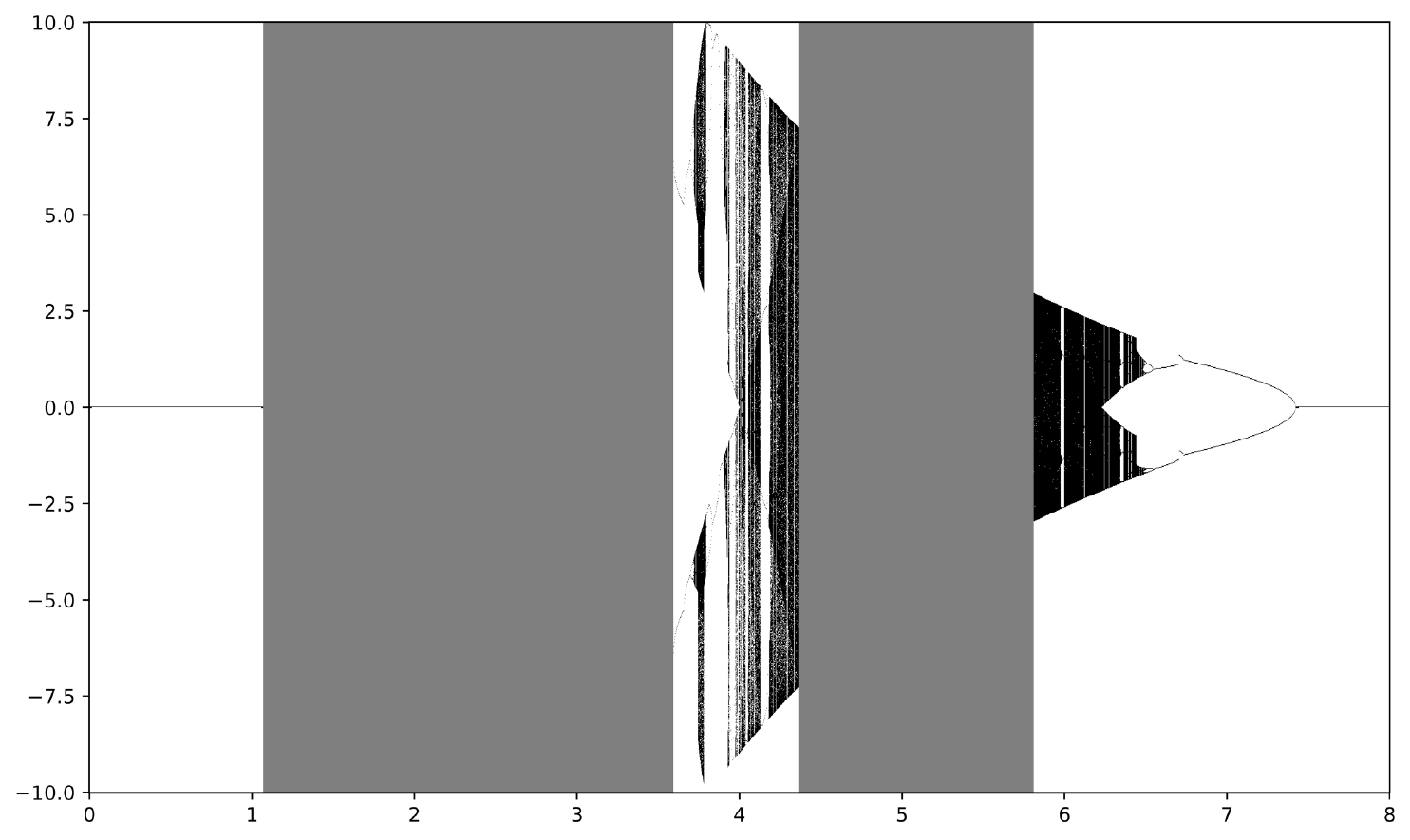}
	\caption{
		Bifurcation diagram for the dynamical system generated by gradient descent on
		the periodic problem~\eqref{eq:exp}, with the learning rate as the varying parameter,
		and with initial point $x=-0.1$. \\
		The diagram represents the accumulation points of the sequence of iterates of
		gradient descent. Grey areas correspond to values of the learning rate for
		which the sequence diverged, axis were broken to hide large areas of divergence.
    Details on the simulation used to create this figure can be found in
    Appendix~\ref{app:bifurcation_diagram}.
	}
    \label{fig:bifurcation}
\end{figure}

While an exhaustive theoretical analysis justifying the bifurcation diagram is out of
reach (even in the well-studied case of the logistic map), we discuss some
specific values of $\eta$ in the following result. There, we put together a collection 
of numerical observations and proofs, providing detailed evidence that gradient descent 
can behave in arbitrarily complex ways on simple time-varying optimization problems. 
\begin{observation} 
  \label{obs:gd_behaviours}
  Consider the trajectories of gradient descent with fixed step size applied to the
  $2$-periodic sequence of functions \eqref{eq:exp}. Depending on the value of $\eta$, the trajectories may have any of the following behaviors:
  \begin{itemize}
    \item[--] For $\eta \leq 1 / 2$ or $\eta \in [7.6, \, 8.4]$, there exists a neighbourhood $\cU$
          of $0$ such that GD converges to $0$ for any initial point $x_0 \in U$; moreover $-0.1 \in \cU$.
    \item[--] For $\eta = 2$ the sequence diverges to $\infty$ for any non-zero initial point.
    \item[--] For $\eta = 3.9$, there exists an $8$-periodic sequence $(c_n)$ and an open set
          $\cU$ containing $0.1$, such that for any $x_0 \in \cU$, we have $(x_n - c_n) \to 0$\,.
    \item[--]
          For $\eta = 6.1$, the system is Li-Yorke chaotic. Precisely, there exists two uncountable
          sets $S_e$ and $S_o$ such that $S_e$ is scrambled by $\bar \Phi_\eta$ and $S_o$ is
          scrambled by $\Phi_{\eta, 2} \circ\Phi_{\eta, 1}$\,.
  \end{itemize}
\end{observation}
We substantiate those observations by combining numerical evidence from simulations with
classical theorems from discrete dynamical systems theory. In particular, we exhibit a
period-$3$ orbit of $\bar \Phi_\eta$, which implies the chaotic behavior by the
Period-Three theorem of \citet{li1975period-t}, by examining the functions $I - \bar
\Phi_\eta$ and $I - (\bar \Phi_\eta)^{3}$.

\begin{proof}{Evidence and proofs substantiating Observation~\ref{obs:gd_behaviours}.}
 \phantom{h} 
  \\
\underline{$\eta \leq 1/2$ or $\eta \in [7.6, 8.4]$}. 
The second derivative of the map $f : x \mapsto \log ( 1 + e^{x^2 / 2} )$ is the function
 $( 1 + e^{-x^2 / 2} + x^2 e^{-x^2/2}) / (1 + e^{-x^2/2})^2$, which takes values in $[1/2, 1.31)$. Therefore
$f_2(x) = f(2 x)$ is $2$-strongly-convex and $(4 \times 1.31)$-smooth; similarly, $f_1(x) =
  f( x / 2)$ is $(1 / 8)$-strongly-convex and $(1.31 / 4)$-smooth. 
   Therefore, using the same derivation as in Example~\ref{example:GD}, for any $n\in \N$,
\begin{equation*}
  \|x_{2n+3}\|^2 \leq ( 1 - 2\eta \mu / 4 + \eta^2  L / 4) \|x_{2n +2 }\|^2
  \leq ( 1 - 2\eta (4 \mu) + \eta^2 (4 L)\big)
  ( 1 - 2\eta  \mu / 4 + \eta^2 L  / 4) \|x_{2n + 1}\|^2 \,,
\end{equation*}
where $\mu= 1/2$ and $L \leq 1.31$. The value of the factor is less than $0.3$ for the
step-size $\eta = 1/2$. Therefore the sequence converges to $0$, regardless of the initial
point.

Note moreover that $0$ is always a fixed point of $\bar \Phi_\eta$, and that $|(\bar
\Phi_\eta)'(0)| < 1$ for any $\eta \in [7.6, 8.4]$. In that case, $0$ is an attracting
fixed point, proving the claim. Simulations show that $-0.1$ belongs to the basin of
attraction of that point.

\underline{$\eta = 2$}.
For $\eta = 2$, we have $|\bar \Phi_\eta(x)| > 2 |x|$ for any $x \in \R\setminus\{0\}$
(see Figure~\ref{fig:eta_2});
therefore the trajectory diverges exponentially fast to $\infty$ from any initial point.
\begin{figure}[h]
  \center
  \includegraphics[width=.5\textwidth]{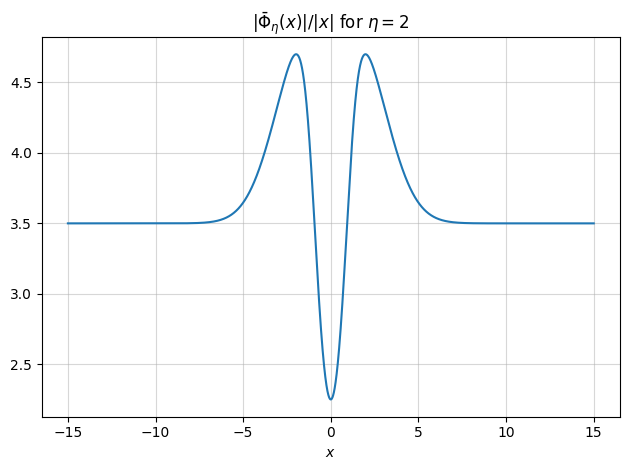}
  \caption{Evidence that $\bar \Phi_\eta(x) > 2 |x|$ for all $x \in \R$, when $\eta=2$.}
  \label{fig:eta_2}
\end{figure}

\underline{$\eta = 3.9$}. 
The map $\psi = (\bar \Phi_\eta)^{(4)}$ admits multiple fixed points; see
Figure~\ref{fig:eta_39}. For example, Newton's method for computing fixed points 
(i.e., the sequence defined by $x_{n+1} = x_n - \psi(x_n) / \psi'(x_n)$), 
initialized at $x=-0.1$, converges numerically
to a point $x_\infty \approx -1.35$; its orbit under $\bar \Phi_\eta$ is the cycle with
approximate values $\{
-1.35, 
5.92, 
-1.57, 
7.04
\}$.
The product of $(\bar \Phi_\eta)'(u)$ over $u$ in this cycle is approximately $-0.26$, 
with absolute value less than $1$. Therefore, by the next result, this cycle is
asymptotically stable for $\bar \Phi_\eta$. 
\begin{theorem}[Prop. 20 Chap. 5 in \citet{Block2006Dynamics}] 
  Let $P = \{x_1, \dots, x_p\}$ be a periodic orbit of a map~$\Phi$, and assume that $\Phi$
  is differentiable at all $x_i \in P$. If $| \Phi'(x_1) \dots \Phi'(x_p)| < 1$, then $P$
  is asymptotically stable.
\end{theorem}
 
This shows that there exists an asymptotically stable $4$-cycle for $\bar \Phi_\eta$,
i.e., that the sequence $(x_{2n+1})$ will converge to that cycle when initialized in the
appropriate (open) set. The result for the whole sequence follows by noting that
$x_{2n+2} = \Phi_{\eta, 1}(x_{2n+1})$ for any $n\in \N$, and that $\Phi_{\eta, 1}$ is
continuous. 
 
\begin{figure}[H]
  \center
  \subfigure[\phantom{x}Plot of $I - (\bar \Phi_\eta)^{(4)}$ for $\eta = 3.9$]{
    \includegraphics[width=.45\textwidth]{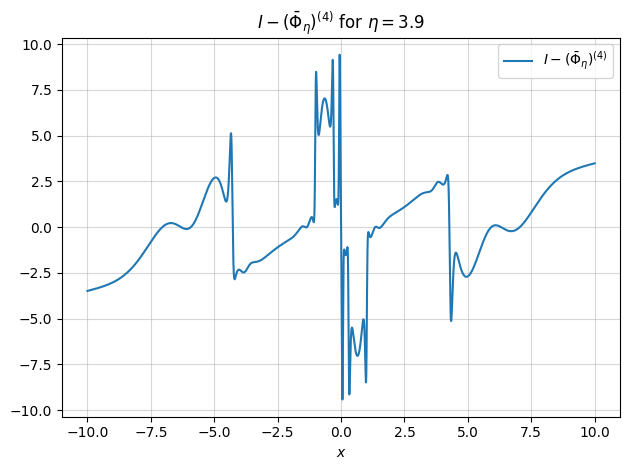}
    \label{fig:eta_39}
  }   
  \quad
  \subfigure[\phantom{x}Plot of $I - (\bar \Phi_\eta)^{(3)}$ for $\eta = 6.1$]{
    \includegraphics[width=.45\textwidth]{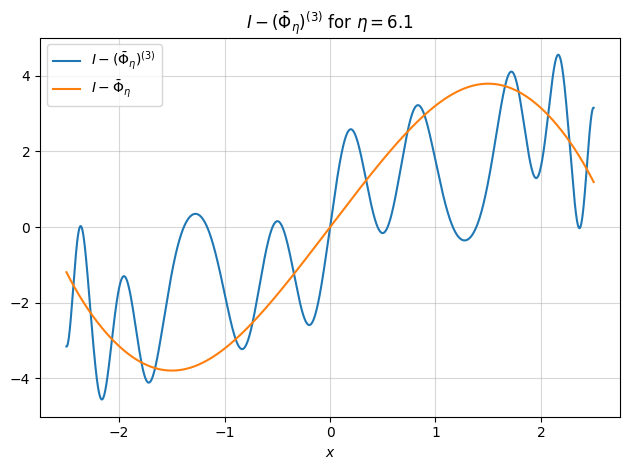}
    \label{fig:eta_61}
  }
\caption{Evidence of the existence of periodic points for $\bar \Phi_\eta$.
Each intersection of the plotted curves with the $x$-axis corresponds to a fixed point of
the iterated function, and thus of a periodic orbit for $\bar \Phi_\eta$.
  A fixed point of $\Phi_\eta^{(3)}$ that is not a fixed point of $\Phi_\eta$ has period
  $3$; the existence of such points implies chaos.}
  \label{fig:fixed_points}
\end{figure}

\underline{$\eta = 6.1$}.
 The map $\bar \Phi_\eta$ is an interval map on $[-2.5, 2.5]$. There exists a $3$-periodic
point $x \approx 0.20$; its orbit is approximately $ \{-0.10, \, 0.20,
\, 0.04\} $. (See Figure~\ref{fig:eta_61}.)
Therefore, by the Period-Three theorem \citep{li1975period-t}, there exists
an uncountable scrambled set $S$ in $[-2.5, 2.5]$ for $\bar \Phi_\eta$.
This corresponds to the odd times of the gradient descent trajectory. The result for
even times follows again by considering the fact $x_{2n+2} = \Phi_{\eta, 1}(x_{2n+1})$, and
that $\Phi_{\eta, 1}$ is continuous.
\qed
\end{proof}

\subsection{Iterated function systems and star-shaped attractors.}\label{sec:starAttractors}
Another interesting phenomenon we observe for time-varying periodic functions is the
convergence to a star-shaped set.
\begin{example}[Periodic time-varying problems with star-shaped limit sets]
  \label{example:starAttractor}
We consider again a $2$-periodic sequence as defined in
\eqref{eq:exp}, here with dimension $d=2$ and
  \begin{align*}
    A_1 =  \begin{bmatrix}
             3/4 & 0 \\
             0           & 5
           \end{bmatrix}\,,
    \qquad A_2 =   \begin{bmatrix}
                     5 & 1           \\
                     1 & 3/4
                   \end{bmatrix}.
  \end{align*}
We illustrate the convergence behavior of this example in Figure~\ref{fig:starAttractor}.
Similarly to the previous example, we observe that gradient descent converges to $0$
for small step-sizes (in Figure~\ref{convergence}, for $\eta \leq 0.45$), then diverges to $\infty$
for larger step-sizes (at $\eta = 0.5$). 
An interesting phenomenon occurs when we further increase the step size: the gradient 
descent iterates start to fill a star-shaped set centered at the origin, 
(see Figure~\ref{star}). That is, a set $\cS$ such that for any point $x \in \cS$, the line segment $[x,0] $ is contained within the set $\cS$. 
These star-shaped sets appear within a relatively small range of step-size choices;
outside this range, we observe divergence.
 
  \begin{figure}[h]
    \centering
    \subfigure[\phantom{x}Star-shaped limit set.]{\includegraphics[scale=0.38]{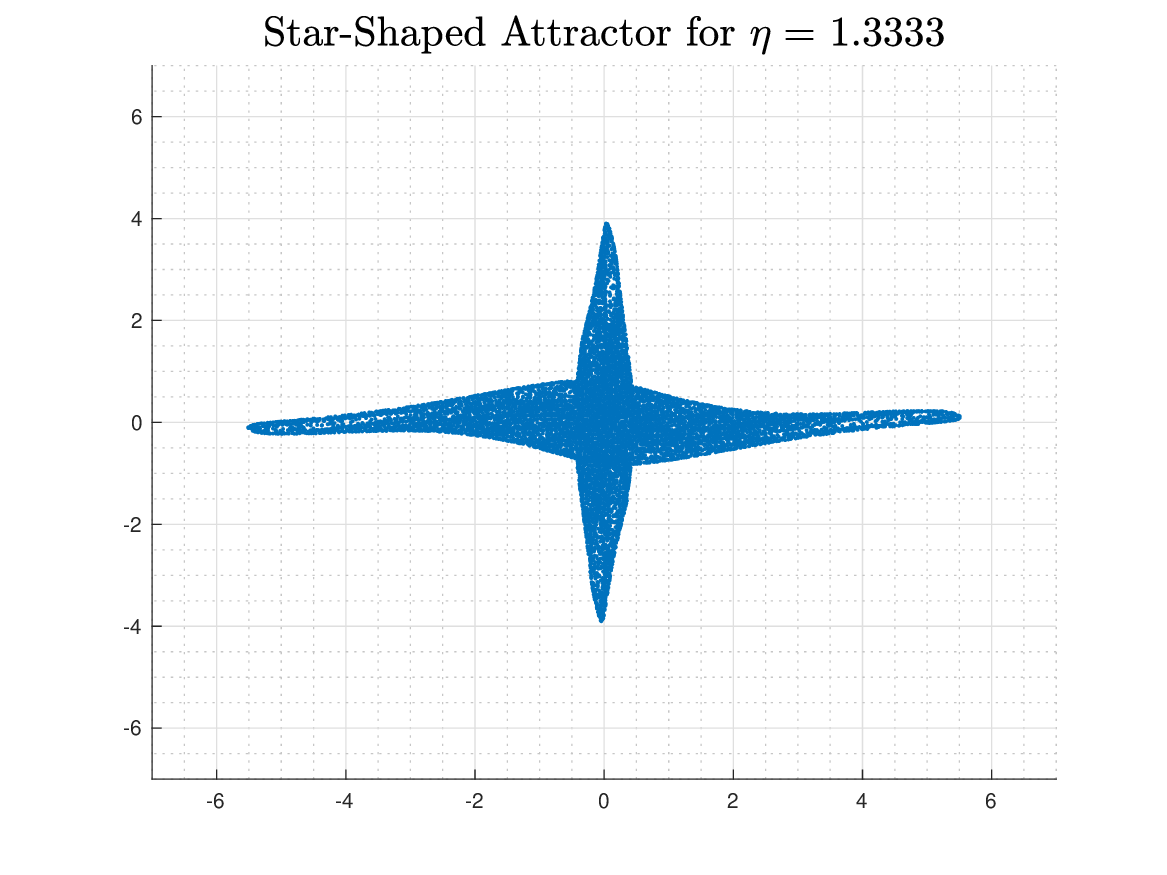}\label{star}}\quad
    \subfigure[\phantom{x}Convergence in norm distance.]{
      \includegraphics[scale=0.38]{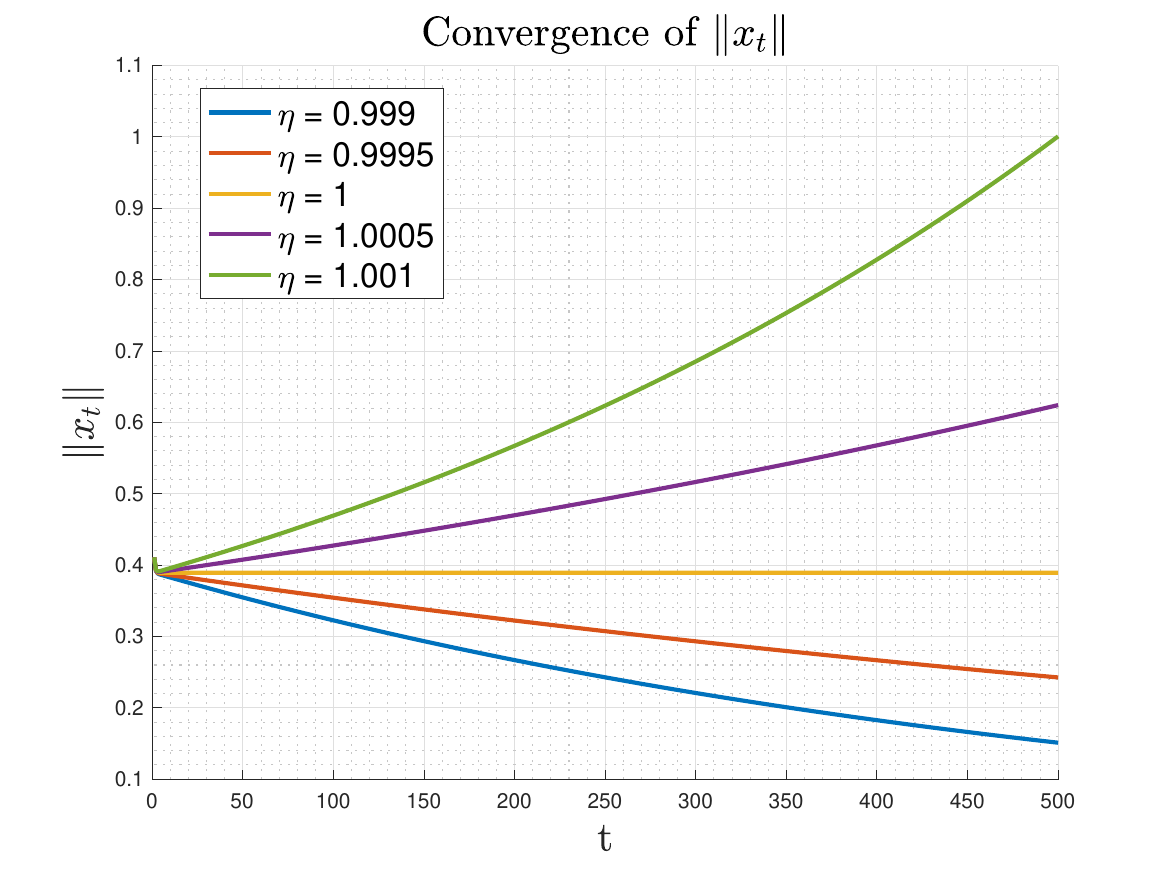}
      \label{convergence}
}  
 \caption{  Illustration of Example \ref{example:starAttractor}. The initial $x_0$ are sampled    uniformly at random from the set $[-500,500]^2$. For Figure \ref{star}, we sample $n =     1000$ initial $x^{(i)}_0, i \in [n]$. The plot shows $x^{(i)}_{n/2}, \dots,     x^{(i)}_{n}$ for each sample $i \in [n]$. Figure \ref{convergence} is based on the $n =     100$ samples for the initial $x^{(i)}_0, i \in [n]$. The plot shows the average convergence over these $n$ samples    $\frac{1}{n}\sum_{i=1}^n \tnorm{x_t^{(i)}}$.    }\label{fig:starAttractor}
  \end{figure}
\end{example}
\smallskip
A complete theoretical analysis of these observations is beyond the scope of this paper. However, we provide a link to existing results in the context of  \emph{ linear} iterated function systems (IFS)
 by \citet{Le_niak_2022} and \citet{BARNSLEY20113124}.    To relate our observations and
these theoretical results, we need to introduce some notation and terminology. As before,
let $\Phi_{\eta,i}:\RR^{d} \rightarrow \RR^d$ denote the gradient descent operator.
 Define the IFS as $\cL = (\RR^{d}; \Phi_{\eta,i}, i \in [k])$
and its one-parametric version $\cL_\tau = (\RR^{d}; \frac{1}{\tau}\Phi_{\eta,i}, i \in
[k])$. That is, the IFS $\cL$ is a collection of $k$ functions  $\Phi_{\eta,i}$ from a (metric) set $\RR^d$ onto itself.  
 We let $\cF^{\Phi} = \{ \Phi_{\eta,i}, i \in [k]\}$ denote the set of all
functions $\Phi_{\eta, i}$. Let $\cR$ denote the set of all non-empty compact subsets of
$\RR^{d}$. Recall the Hutchinson operator: $\cF:\cR \rightarrow \cR$, with \[\cF(\cB) = \bigcup_{i \in [k]} \Phi_{\eta,i}(\cB).\]
A non-empty set $\cA \subset\RR^{d}$ is an attractor of the IFS if
\begin{align*}
  \cF(\cA) = \cA\qquad \text{and}\qquad\forall \cB \in \cR \quad  \lim_{t \rightarrow \infty} \cF^{(t)}(\cB) = \cA,
\end{align*}
where $\cF^{(t)}$ denotes the $t$-fold composition of the Hutchinson operator $\cF$. The limit is taken with respect to the Hausdorff distance. Note
that the limit set of the GD operator $\bar \Phi_{\eta}$ (recall~\eqref{def:GDoperator}) for periodic time-varying functions is a subset of the attractor of the IFS $\cL = (\RR^{d}; \Phi_{\eta,i}, i \in [k])$. 
 Further, denote by $\rho(\cF^{\Phi})$, the joint spectral radius of the IFS (see e.g.
\citet{Rota1960ANO}). In \citet[Theorem 2]{BARNSLEY20113124} and \citet{Le_niak_2022}, it
is shown that an irreducible one-parametric linear IFS $\cL_{\tau}$  has a star-shaped, central-symmetric attractor if
$\tau = \rho(\cF^{\Phi})$.  We want to emphasize that these results only apply to \emph{linear} IFS, while our observations are due to non-linear IFS. Hence, the theory provided in  \citet{BARNSLEY20113124} and \citet{Le_niak_2022} does not explain our observation, however, it indicates an interesting future research direction and reveals an interesting link between time-varying optimization problems and IFS. 

\section{Discussion, Limitations and Future Work}\label{sec:limitations}
 
We studied the problem of tracking equilibria in time-varying variational inequality problems. Positive results include sublinear tracking error bounds for tame instances and contractive algorithms. These bounds are proven to be tight for this class of algorithms. Other positive results concern the case of periodic instances with strongly monotone operators, where we propose meta-algorithms for aggregating multiple experts. These algorithms result on logarithmic or even constant tracking bounds. Finally, we explore gradient-type methods with constant stepsize for periodic instances, as an agnostic approach to tackle periodic instances with unknown periods. Our results show an intricate landscape where convergence, chaos, and nonconvergence may arise, making the use of these algorithms difficult even in simple cases.

 We discuss next some of the limitations in this work.
The algorithm of Section~\ref{sec:metaPeriodic} is computationally more efficient than
the algorithm in Section~\ref{sec:unconstrained_lip}. However, in contrast to the latter,
it requires a bounded domain and provides (only) logarithmic guarantees. It would be an
interesting question for future work to combine the merits of both approaches.
Another limitation becomes clear when considering the special case of games: In this
case, the aggregation schemes discussed in this section may be
unnatural, as they require coordination between the players. Decoupling the aggregation
in general cases would be an interesting direction for further investigation.

We also remark that both results in Section~\ref{sec:metaPeriodic} and
\ref{sec:unconstrained_lip} require deterministic feedback.
An interesting direction is to consider more robust models, in which we would
observe the inexact values of the operators. Examples of such models could be a
stochastic feedback model or time-varying VI problems that admit periodic behavior only in the
limit.
These feedback models would capture application examples, such as fluctuating markets,
in a more realistic setting. Inexact feedback models are inherently more challenging; the convergence guarantees in this setting are weaker, and the aggregation step would be more difficult to
implement.

 Another interesting future research direction emerges from the relation of time-varying games to stochastic games. Stochastic games are a broad class of time-varying games, where the state of the game evolves depending on both the players' actions as well as random events. This setting contains both tame time-varying games and periodic ones, but it is a much broader class. Understanding the class of stochastic games for which tracking is attainable is then a very interesting and important direction for future research.


 \section*{Acknowledgments.}
We thank Tim van Erven for his helpful and insightful comments on an earlier version of
this manuscript. We also thank Panayotis Mertikopoulos for his feedback on the related
literature for Section \ref{sec:contraction}. Further, we thank Mathias Staudigl for
some discussions on time-varying games. Work was done while Sachs was supported by the
Netherlands
Organization for Scientific Research (NWO) under grant number
VI.Vidi.192.095, now Sachs is supported by the European Union - Next Generation EU, PRIN 2022, CUP:J53D23007170001. Guzm\'an's research was partially supported by INRIA Associate Teams project, ANID FONDECYT 1210362 grant and ANID Anillo ACT210005 grant.
 
 
\bibliography{cycling} 


\ \\

\appendix 

\section{Proofs for Section \ref{sec:introduction}.}\label{sec:appendixIntro}
\begin{lemma}[Lower Bound]
\label{LB:tracking}
For any algorithm $\cA$, there exists a sequence $(F_t)_{t \in [T]}$ such that
$
\track(\cA) \geq \frac{1}{4}\sum_{t=1}^T \norm{Z_t^\star - Z_{t+1}^\star}^2
$
and
$
\sum_{t=1}^T \norm{Z_t^\star - Z_{t+1}^\star}^2 = \frac{1}{4}D^2T
$.
\end{lemma}
\begin{proof}{Proof.}
We show this lower bound by defining a sequence $(Z_t^\star)_{t \in [T]}$ with $Z^\star_t
\in \cZ$. The claim follows by choosing $F_t$ as the gradient of a convex function $f_t$
with a unique minimizer at  $Z_t^\star$. This implies that $Z_t^\star$ solves the
\eqref{VIP}. We assume that $\cZ = [-1,1] \subseteq \RR$ and choose $Z_t^\star \in
\{-1,0,1\}$. 
 For any $t\in \NN$, we distinguish two cases: $Z_t \geq 0$ and $Z_t< 0$. First assume
$Z_t \geq 0$. We distinguish again two cases: (1) $Z_{t-1}^\star \in \{0,1\}$ and (2)
$Z_{t-1}^\star = -1$. For case (1), we set $Z_t^\star = -1$ and directly obtain $\norm{Z_t
- Z_t^\star}^2 \geq \frac{1}{4} \norm{Z_{t-1}^\star - Z_t^\star}^2$. For case (2), if $Z_t
\geq \frac{1}{2}$, we set $Z_t^\star = 0$ otherwise $Z_t^\star = 1$. In both cases, we
directly obtain $\norm{Z_t - Z_t^\star}^2 \geq \frac{1}{4} \norm{Z_{t-1}^\star -
Z_t^\star}^2$. 
We build $Z_{t-1}^\star$ using a symmetric construction when $Z_t <0$, which
completes the proof. 
\qed
\end{proof}

\section{Proofs for Section \ref{sec:contraction}.} 
\subsection{Proofs: Example \ref{sec:examplesContraction}}
 \label{appendix:generality}
We adopt an example by \citet{Karimi2016LinearCO}.
 Recall the time-varying game from \eqref{nonmonoGame} with $a \in [0,1]$
\begin{align*}
  \ell(x,y) = x^2 + 3 \sin^2(x) + a \sin^2(x) \sin^2(y) - y^2 - 3 \sin^2(y) .
\end{align*}
This game is not convex-concave. Its gradients with respect to $x$ are for any $x, y$, 
\begin{align*}
  \nabla_x \ell(x,y) &= x + 3 \sin(x)\cos(x) + a \sin(x) \cos(x) \sin^2(y) = 2x +  \sin(2x) (3 + a \sin^2(y) ) / 2 \, , \\
  \nabla_y \ell(x,y) &= - y - 3 \sin(y)\cos(y) + a \sin(y) \cos(y) \sin^2(x) = - 2y - \sin(2x) (3 - a \sin^2(x)) / 2  \, . 
\end{align*}
 
We now show that the game defined in \eqref{nonmonoGame} is
\begin{enumerate}
  \item not a monotone game,
  \item satisfies the Restricted Secant Inequality (RSI) in $x$ for any fixed $y$ (and in $y$ for any fixed $x$, respectively), i.e., there exists a $\mu_{\rm RSI}>0$ such that \[ \forall x \in \RR: \; \dprod{\nabla_x\ell(x,y),x - x^\star} \geq \mu_{\rm RSI}\norm{x - x^\star}^2 ,\]
  where $x^\star$ denotes the projection of $x$ onto the set of minimizers. 
  \item GDA satisfies the contraction assumption.
\end{enumerate}
\begin{proposition}\label{prop:RSI}
For any $x, y \in \RR$, the functions $\ell(\, \cdot\, ,y)$ and $\ell(x, \, \cdot\,)$
satisfy the RSI with $\mu = 1/4$.
\end{proposition}
\begin{proof}{Proof.}
We want to show that
\begin{align}
\label{ineq:RSI}
  \big(2x +  \sin(2x) (3 + a \sin^2(y)) / 2\big)(x - x^\star) \geq \mu (x-x^\star)^2.
\end{align}
First note that for any $y$, $\ell(x,y)$ has a unique minimizer $x^\star = 0$. Further, for $x=0$, \eqref{ineq:RSI} holds for any $\mu $. Hence, assume $x \neq 0$.  Then \eqref{ineq:RSI} is satisfied if
\[ 
  2 +  \frac{\sin(2x)}{2x} (3 + a \sin^2(y))  \geq \mu ,
\]
 
Using the fact that $\sin(x) \in [-1,1]$ and that the minimum of $\sin(x)/x$ is attained
at $\{-\frac{3}{2}\pi, \frac{3}{2}\pi\}$  (see, e.g., \citet{Klen:2010aa}), 
$
  \sin(2x) / (2x) \geq - 0.22 .
$
Therefore, since $3 \leq (3 + a \sin^2(y)) \leq 4$ for any $a$ and $y$, we have
\[
  \frac{\sin(2x)}{2x} (3 + a \sin^2(y)) \geq - 0.88 \, , 
  \quad 
  \text{so}
  \quad 
  2 + \frac{\sin(2x)}{2x} (3 + a \sin^2(y)) \geq 1.12, 
\]
proving the claim.
\qed
\end{proof}
\begin{proposition}
\label{prop:nonConvex}
  The game defined in \eqref{nonmonoGame} is not convex and a fortiori not strongly convex. 
\end{proposition}
\begin{proof}{Proof.}
  The second derivative with respect to $x$ evaluated at $x, y$ is 
  $
    2 + \cos(2x)(3 + a \sin^2(y)) \, .
  $
  which can be negative, e.g.\,at $x = \pi / 2$ and $y=0$.
\qed
\end{proof}

\begin{proposition}
\label{prop:contraction}
Contraction is satisfied for problem \eqref{nonmonoGame} with GDA with $C = (1-\mu^2_{\rm RSI} / L^2)^{1 / 2}$ where $\mu_{\rm RSI}$ denotes the RSI parameter. 
\end{proposition}
\begin{proof}{Proof.}
As for the strongly-convex case:
\begin{align*}
   \|x_2 - x^\star\|^2
    = \| x - x^\star - \eta \nabla_x \ell(x,y)\|^2
    = \|x - x^\star\|^2
    -2\eta \langle x - x^\star, \nabla_x \ell(x,y)  \rangle
    + \eta^2\|\nabla_x \ell(x,y)\|^2.
\end{align*}
By RSI
\begin{align*}
 \langle x-x^\star, \nabla_x \ell(x, y)  \rangle
        + \langle y-y^\star, -\nabla_y \ell(x, y)  \rangle 
        \geq \mu_{\rm RSI}\|x-x^\star\|^2 + \mu_{\rm RSI}\|y-y^\star\|^2,
\end{align*}
thus,
\begin{align*}
 \|x_2 - x^\star\|^2 + \|y_2 - y^\star\|^2
        \leq ( 1- 2 \eta  \mu_{\rm RSI})(\|x_1 - x^\star\|^2 + \|y_1 - y^\star\|^2)
        + \eta^2 \big(\|\nabla_x \ell(x,y)\|^2 + \|\nabla_y \ell(x,y)\|^2\big).
\end{align*}
The result follows as before via $L$-smoothness.
\qed
\end{proof}

\subsection{Tightness.}
\label{app:tightness}

\paragraph{Proof of Theorem~\ref{thm:tightness}.}
\begin{proof}{Proof.}
We now provide an example of a sequence of problems, with a $C$-contractive algorithm,
such that the bound given in Theorem~\ref{thm:contraction2track} is tight. Our example is
an elementary analysis of gradient descent on a sequence of $1$-dimensional quadratic
optimization problems, where the optima drift linearly.
Consider a sequence of quadratic functions over $\R$ defined by 
\[
f_t(x) = \frac{(x - c_t)^2}{2} \,; 
\]
for our example, let us assume that $(c_t)$ forms an arithmetic progression with common
difference $-b$, i.e., $c_{t+1} = c_t - b$ for all $t\geq 1$. Consider also the gradient
descent algorithm started at $x_1 \in \R$, and assume furthermore that $x_1 - c_1 \geq b
/ (1-C)$.

Then for any $t \in \R$, gradient descent started at $x_1 \in \R$ with step-size $(1-C)
\in (0, 1]$ is a $C$-contraction since for any $t \geq 1$, the update rule is
\[
x_{t+1} - c_t = x_t - (1-C) (x_t - c_t) - c_t = C (x_t - c_t) \,.
\]
Let us introduce the short-hand notation $d_t = x_t - c_t$. Then for any $t \geq 1$, 
\[
d_{t+1} = x_{t+1} - c_{t+1} 
= x_{t+1} - c_t + c_t - c_{t+1}
= Cd_t + b
\]
The sequence $(d_t)$ therefore forms an arithmetico-geometric progression and
so for any $t\geq 1$, using our assumption that $x_1 - c_1 \geq b / (1 - C)$,
\[
d_t - \frac{b}{1-C} = C^{t-1} \bigg(d_1 - \frac{b}{1-C}\bigg)
\geq 0
\,.
\]
Therefore, for any $t \geq 1$, the squared-error $d_t^2$ can be lower bounded as
\begin{align*}
d_t^2 &= 
\bigg(d_t - \frac{b}{1 - C} + \frac{b}{1 - C}\bigg)^2 \\
&= \bigg(d_t - \frac{b}{1 - C}\bigg)^2
+ \frac{b^2}{(1-C)^2}
+ 2 \bigg( d_t -\frac{b}{1 - C}\bigg)\frac{b}{1 - C} 
 \geq \frac{b^2}{(1-C)^2}
\,.
\end{align*}
Conclude by summing the inequality above over $t \in [T]$,
\[
\sum_{t=1}^{T} (z_t - c_t)^2 
\geq 
\frac{Tb^2}{(1-C)^2} \geq \frac{1}{(1-C)^2} \sum_{t=2}^T |c_{t}- c_{t-1}|^2 \,, 
\]
proving the claim. 
\qed
\end{proof}

\section{Proofs of main positive results.}

\subsection{Proof of Theorem \ref{thm:aggregation}}
\label{sec:appendix:proofThm3}
 
\begin{proof}{Proof.}
The setting for aggregation we consider here is an instance of Prediction with Expert
Advice, see \citet[Chapter 2]{cesa-bianchi2006predicti}, where the $K$ experts provide
the advice $Z_t^{(i)}$, the outcomes are the pairs $(F_t(z_t), z_t)$, and where the loss
is the function over {$\cZ$ and the dual space $(\cZ)^\star$, i.e., $\cL:\cZ \times ((\cZ)^\star \times \cZ)\rightarrow \RR$}:
\[
  \cL(a ; g, z) = \langle g, a \rangle + \frac{\mu}{2} \| a - z\|^2
  = \frac{\mu}{2} \Big\|a - z - \frac{1}{\mu} g \Big\|^2 - \frac{1}{\mu} \|g\|^2
  +  \langle z, g \rangle \,;
\]
we apply it to $(g, z) = (F_t(z_t), z_t)$, thus the predictions will be bounded by $D$
and the outcome vectors by $D + G / \mu$. This loss function, which is an instance of
the square loss, is $\lambda$--mixable for the mean (or, in other words,
$\lambda$--exp-concave)
for $\lambda \leq (4\mu(D + G/\mu)^2)^{-1}$, as long as $\|x + (1 / \mu) g\| \leq D + G / \mu$
and $\|a\| \leq D$. Playing the Exponential Weights updates from \eqref{eq:exp_weights}
in this context is equivalent to using Vovk's aggregating algorithm
\cite{vovk1995a-game-o}, which guarantees for $\lambda$--mixable losses that
 \[
  \sum_{t=1}^T \cL(Z_t ; F_t(Z_t), Z_t)
  \leq
  \min_{i \in [K]}
  \sum_{t=1}^T \cL(Z^{(i)}_t ; F_t(Z_t), Z_t)
  + \frac{\log K}{\lambda} \,. 
\]
This can be written as
\begin{align}\label{ineq:VovkRewritten}
  \sum_{t=1}^T
  \langle F_t(Z_t), Z_t \rangle
  \leq
  \min_{i \in [K]}
  \sum_{t=1}^T
  \langle F_t(Z_t), Z_t^{(i)}\rangle
  + \frac{\mu}{2} \| Z_t^{(i)} - Z_t\|^2
  + \frac{\log K}{\lambda} \,.
\end{align}
Now note that for any $u \in \cZ$,
\begin{align*}
  \langle \tilde F_t(Z_t^{(i)}), Z_t^{(i)} - u \rangle
  - \frac{\mu}{2} \| Z_t^{(i)} - u\|^2
   & =
  \langle F_t(Z_t) + \mu(Z_t^{(i)} - Z_t), Z_t^{(i)} - u \rangle
  - \frac{\mu}{2} \| Z_t^{(i)} - u\|^2 \\
   & =
  \langle F_t(Z_t), Z_t^{(i)} - u \rangle
  +\mu \langle Z_t^{(i)} - Z_t, Z_t^{(i)} - u \rangle
  - \frac{\mu}{2} \| Z_t^{(i)} - u\|^2 \\
   & =
  \langle F_t(Z_t), Z_t^{(i)} - u \rangle
  - \frac{\mu}{2} \| Z_t - u\|^2
  + \frac{\mu}{2} \| Z_t^{(i)} - Z_t\|^2 \,.
\end{align*}
Therefore, plugging this bound in \eqref{ineq:VovkRewritten}, we see that for any sequence of
comparators $(C_t) \in \cZ^T$,
\begin{align*}
  \sum_{t=1}^T
  \langle F_t(Z_t), Z_t - C_t \rangle
  - \frac{\mu}{2} \| Z_t - C_t\|^2
   & \leq
  \min_{i \in [K]}
  \sum_{t=1}^T
  \langle F_t(Z_t), Z_t^{(i)} - C_t\rangle
  - \frac{\mu}{2} \| Z_t - C_t\|^2
  + \frac{\mu}{2} \| Z_t^{(i)} - Z_t\|^2
  + \frac{\log K}{\lambda} \\
   & =
  \min_{i \in [k]}
  \sum_{t=1}^T
  \langle \tilde F_t(Z_t^{(i)}), Z_t^{(i)} - C_t \rangle
  - \frac{\mu}{2} \| Z_t^{(i)} - C_t\|^2
  + \frac{\log K}{\lambda} \,.
\end{align*}
If in addition, $(C_t)$ is $k$-periodic, then, by applying the guarantees from
Corollary~\ref{cor:periodic_gd} to the operators $\tilde F_t$ which are upper bounded in norm by
$G + \mu D$, we obtain
\[
  \min_{i \in [k]}
  \sum_{t=1}^T
  \langle \tilde F_t(Z_t^{(i)}), Z_t^{(i)} - C_t \rangle
  - \frac{\mu}{2} \| Z_t^{(i)} - C_t\|^2
  \leq
  \frac{k(G + \mu D)^2}{2\mu} \log\Big( \frac{T}{k} + 1\Big) \,.
\]
Replace $\lambda$ by its value to conclude the proof.
\qed
\end{proof}

\subsection{Proof of Theorem \ref{thm:constRegretAdaHedge}} 
\label{app:proof_of_ada}
We separate the proof into three distinct lemmas. We start with an analysis of the
aggregation method, showing that the regret of the full procedure can be controlled by
the cumulative distance of the correctly tuned base algorithm to the comparators, plus the
cumulative mixability gap. We then proceed to bound the said cumulative mixability gap, 
by showing that the instantaneous mixability gaps are eventually $0$, making the sum 
bounded independently of $T$. (More accurately, the dependence is only in $\what D_T$ and
$\what G_T$, which will then be bounded independently.)

Recall that $(x)_+$ denotes the positive part of a real number $x$, and that
\[
\what D_T = \max_{{t \in [T], \,  i, j \in [K]}} \|Z_t^{(i)} - Z_t^{(j)}\|
\quad
\text{and}
\quad
\what G_T = \max_{t \in [T], \,  i \in [K]} \|F_t(Z_t^{(i)} )\| \,.
\]
\begin{lemma}\label{lem:agg_ada}
For any sequence of comparators $(C_t)$, the meta-algorithm 
with adaptive learning rate and losses as defined in \eqref{eq:ew_and_losses} and \eqref{eq:ada_lr} ensures that
\begin{align*}
\sum_{t=1}^{T}& \left(\dprod{ F_t(Z_t), Z_t -C_t }
- \frac{\mu}{2} \| Z_t - C_t\|^2\right) 
\leq 2 \sum_{t=1}^{T} \big(\bar \ell_t - m_t\big)_+ 
+ \left(\what G_T  + \mu \what D_T \right)  
\sum_{t=1}^T  \big\|Z^{(k)}_t - C_t \big\| \,.
\end{align*}

\end{lemma}
\begin{proof}{Proof.}
We start by applying the standard exponential weight analysis, which relates the
cumulative mix loss to the minimal cumulative loss among the $K$ base algorithms.
(We refer the reader to Lemma 1, paragraph 3 and Lemma 2 in \citet{De-Rooij:2014aa} 
for a proof.)
This yields
\[
  \sum_{t=1}^{T} m_s \leq 
  \min_{i \in [K]} \sum_{t=1}^{T} \ell_{s, i} + \frac{\log K}{\lambda_T} \,.
\]
Therefore, using the fact that the $ u \leq u_+$ for any $u \in \R$, 
\begin{align}
\label{bound:mixlossLogK}
  \sum_{t=1}^{T} \bar \ell_t
  =
  \sum_{t=1}^{T} (\bar \ell_t - m_t)
  + 
  \sum_{t=1}^{T} m_t
  \leq 
  \sum_{t=1}^{T} (\bar \ell_t - m_t)_+
  + 
  \min_{i \in [K]} \sum_{t=1}^{T} \ell_{t, i} + \frac{\log K}{\lambda_T} \,.
\end{align}
Next we use that by definition $(\log K )/ \lambda_T = \sum_{t=1}^{T-1} (\bar \ell_t -
  m_t)_+$. Plugging this into \eqref{bound:mixlossLogK} and rearranging the terms gives
\[
  \sum_{t=1}^{T} \bar \ell_t - \min_{i \in [K]} \sum_{t=1}^{T} \ell_{t, i} 
  \leq 
   2 \sum_{t=1}^{T} (\bar \ell_t - m_t)_+
  \,.
\]
Replacing $\bar \ell_t$ and $\ell_{t,i}$ by their definitions in terms of $F_t$ and $Z_t^{(i)}$, we obtain
\begin{align*}
\sum_{t=1}^{T} \langle F_t(Z_t), Z_t\rangle 
  - \min_{i \in [K]}  
 \sum_{t=1}^{T}\left( \langle F_t(Z_t), Z^{(i)}_t\rangle  
 + \frac{\mu}{2}\|Z^{(i)}_t - Z_t\|^2\right)
  \leq 
   2 \sum_{t=1}^{T} (\bar \ell_t - m_t)_+
  \,.
\end{align*}
Let us now rearrange the terms and subtract $\sum_{t=1}^T\frac{\mu}{2}\norm{Z_t - C_t}^2$
on both sides to obtain
\begin{align*}
\sum_{t=1}^{T}& \left(\dprod{ F_t(Z_t), Z_t -C_t }
- \frac{\mu}{2} \| Z_t - C_t\|^2\right) \\
  &\leq \min_{i \in [K]}  
 \sum_{t=1}^{T}\left( \dprod{ F_t(Z_t), Z^{(i)}_t - C_t  } 
 + \frac{\mu}{2}\big \|Z^{(i)}_t - Z_t \big\|^2 - \frac{\mu}{2}\norm{ Z_t - C_t}^2 \right)
 + 2 \sum_{t=1}^{T} \big(\bar \ell_t - m_t\big)_+\\
   &\leq    
 \sum_{t=1}^{T}\left( \dprod{ F_t(Z_t), Z^{(k)}_t - C_t  } 
 + \frac{\mu}{2}\big\|Z^{(k)}_t - Z_t \big\|^2  -\frac{\mu}{2}  \norm{ Z_t - C_t}^2 \right)
  + 
   2 \sum_{t=1}^{T} \big(\bar \ell_t - m_t\big)_+
  \,.
\end{align*}
Finally, for any $t$, the loss of the $k$-th base algorithm can be controlled by 
the distance to the comparators: 
\begin{multline*}
  \quad
  \langle F_t(Z_t), Z^{(k)}_t - C_t \rangle  
 + \frac{\mu}{2}\|Z^{(k)}_t - Z_t\|^2 -\frac{\mu}{2}\| Z_t - C_t\|^2  \\
  = \langle F_t(Z_t), Z^{(k)}_t - C_t \rangle  
 + \frac{\mu}{2} \big\langle  Z^{(k)}_t + C_t - 2 Z_t, Z^{(k)}_t - C_t \big\rangle 
 \leq \big(\what G_T  + \mu \what D_T \big) \|Z^{(k)}_t - C_t\| \,, 
\end{multline*}
where we used the Cauchy-Schwarz inequality to conclude. Summing over $t$ gives the
claim. 
\qed
\end{proof}
\smallskip

We now proceed to bound the cumulative mixability gap separately. The proof is essentially
an application of the mixability of the square loss with respect to the mean, which is
equivalent to its exp-concavity. Recall a function $f$ is said to be
$\alpha$--exp-concave if $\exp(-\alpha f)$ is convex. For completeness, we provide a
self-contained proof of the following result, with no explicit call to mixability.
\begin{lemma}
  \label{lem:mix_gap_ada}
Under the assumptions of Lemma~\ref{lem:agg_ada}, 
\[
  \sum_{t=1}^{T} (\bar \ell_t - m_t)_+
  \leq
 \frac{2(\mu \hat D_T + \what G_T)^2}{\mu} \log K 
  + \what G_T \what D_T + \frac{\mu}{2} \what D_T^2 
  \,.
\]
\end{lemma}
\begin{proof}{Proof.}
We prove the claim by showing that the summand in the statement is eventually zero.
First, note that the sequence $(\lambda_t)$ is non-increasing. We define the threshold
value
\[
  \lambda^\star =\frac{\mu}{2(\mu \what D_T + \what G_T)^2}\, , 
\]
and let $t_0 \in \NN$ denote the largest integer smaller than $T$ such that $\lambda_t
  \geq \lambda^\star$. Let us first show that
\begin{align}
\label{ineq:firstboundtZero}
\sum_{t=1}^{t_0} (\bar \ell_t - m_t)_+ 
 \leq \frac{2(\mu \what D_T + \what G_T)^2}{\mu} \log K 
 + \what G_T \what D_T + \frac{\mu}{2} \what D_T^2 \, ,
\end{align}
and that for any $t \geq t_0+1$, we have $(\bar\ell_t - m_t)_+ = 0$. 
To show \eqref{ineq:firstboundtZero}, we note that by definition
\begin{align*}
\lambda_{t_0} 
= \frac{\log K}{\sum_{t=1}^{t_0-1} (\bar \ell_t - m_t)_+} \geq \lambda^\star \, , 
\quad \text{so}
\quad \sum_{t=1}^{t_0-1} (\bar \ell_t - m_t)_+ 
\leq \frac{\log K}{\lambda^\star}  =  \frac{2(\mu \what D_T + \what G_T)^2}{\mu} \log K \,.
\end{align*}
Furthermore, for any $t \in [T]$, 
\begin{align}
\label{ineq:generalBoundonLM}
( \bar \ell_t - m_t )_+\leq \absv{ \bar \ell_t - m_t }
 \leq \Big|\max_{i \in [K]} \ell_{i, t} - \min_{j \in [K]} \ell_{j, t} \Big|
 \leq  \what G_T \what D_T + \mu \what D_T^2  \,.
\end{align}
Hence, applying this inequality in particular to $t_0$, we obtain, 
\begin{align}
\label{bound:sum2tzero}
\sum_{t=1}^{t_0} (\bar \ell_t - m_t)_+ 
= (\bar \ell_{t_0} - m_{t_0})_+ 
+ \sum_{t=1}^{t_0-1} (\bar \ell_t - m_t)_+ 
\leq \frac{2(\mu \what D_T + \what G_T)^2}{\mu} \log K 
+ \what G_T \what D_T + \mu \what D_T^2 \, .
\end{align}
Let us now consider $t > t_0$, and let us show that $ \bar \ell_t \leq m_t$. We prove this
by using the exp-concavity of the square loss. 
Note first that the definition of $t_0$ implies that $\lambda_t < \lambda^\star$, since
the sequence $(\lambda_t)$ is non-increasing.
Denote by $B(Z_t, \what D_T)$ the Euclidean ball centered at $Z_t$ of radius $\what D_T$, that is $B(Z_t,\what D_T) = \{ Z \in \cZ: \norm{Z-Z_t} \leq D_T\}$.
The map
\begin{align*}
  h_t :B(Z_t, \what D_T) & \longrightarrow \R \\
  u \quad & \longmapsto \langle F_t(Z_t), u \rangle + \frac{\mu}{2} \|Z_t - u \|^2
\end{align*}
is $\lambda^\star$--exp-concave, and therefore $\lambda_t$--exp-concave since $\lambda_t
  \leq \lambda^\star$. Indeed, for any
$u \in B(Z_t, \what D_T)$
\begin{align*}
 \lambda^\star \nabla_u h_t(u) \nabla_u h_t(u)^\top 
 &= \lambda^\star (F(Z_t) + \mu (u - Z_t))(F(Z_t) + \mu (u - Z_t))^\top\\
  &\preccurlyeq 
  \lambda^\star \norm{F(Z_t) + \mu (u - Z_t)}^2 I_d\\
  &\preccurlyeq 
  \lambda^\star 2 (\what G_T^2 + \mu \what D_T^2) I_d\\
  &\preccurlyeq 
  \mu I_d = \nabla^2_u h_t(u)\,, 
\end{align*}
which is exactly the condition for $\lambda^\star$--exp-concavity (see
\citet{Hazan:2022aa}, Section 4.2). Furthermore, we note that for any base
algorithm $i \in [K]$, we have $Z_t^{(i)} \in B(Z_t, \what D_T)$. Hence, by
$\lambda_t$-exp-concavity
\begin{align}
\label{ineq:Expcon}
  \sum_{i = 1}^{K} p_{t, i} e^{- \lambda_t h_t(Z_t^{(i)})}
  \leq \exp \bigg(- \lambda_t h_t \bigg(\sum_{i = 1}^{K} p_{t, i} Z_t^{(i)}\bigg)  \bigg)
  = \exp \big(- \lambda_t h_t (Z_t) \big)\nonumber 
\end{align}
which we rearrange into
\[
  \bar \ell_t = \langle F_t(Z_t), Z_t\rangle   = h_t(Z_t)
  \leq -\frac{1}{\lambda_t} 
  \ln \sum_{i =1}^{K} p_{t, i} e^{-\lambda_t \, \ell_{i, t}}  = m_t \,.
\]
We demonstrated that $\lambda_t \leq \lambda^\star$ implies $ (\bar \ell_t - m_t)_+ = 0$ which,
together with \eqref{bound:sum2tzero}, completes the proof.
 \qed
\end{proof}
\smallskip

The next lemma guarantees that all the base algorithms stay within a bounded distance from the set of solutions of the operators in the sequence. It is a straightforward 
consequence of the fact that the forward{-backward} algorithm is contractive in that setting. 
 
\begin{lemma}
  \label{lem:bounded_d_and_g}
If the sequence of operators is strongly monotone and Lipschitz, if $\cZ = \R^d$, and
base algorithms are the cyclic forward{-backward} algorithm, then 
  \[
    \what D_T \leq   2\sqrt{2 \kappa^2 + 1}D_0 
    \quad 
    \text{and}
    \quad
    \what G_T \leq  L\sqrt{2 \kappa^2 + 1} D_0  \,.
  \]
\end{lemma}

\begin{proof}{Proof.}
Remember that $\eq$ denotes the set of solutions of the operators $(F_t)$. We first
show that due to the contraction of the forward{-backward} algorithm (cf.,
Section~\ref{sec:examplesContraction}, Example~\ref{example:GD}), for all $t \in [T]$ and
base algorithm $i \in [K]$, there exists a point $Z^\star_t \in \eq$ such that
$\|Z_t^{(i)} - Z^\star_t \| \leq \kappa D_0 + D_c$, where $D_c = \max_{C,C' \in
    \eq}\norm{C-C'}$. Indeed, for all base algorithms $i\in[K]$
 \begin{align*}
 \big\|Z^{(i)}_t - Z_{t}^\star\big\|^2 
 \leq (1-\kappa^{-1})\big\|Z^{(i)}_{t-1} - Z_{t-1}^\star\big\|^2
+ \kappa \norm{Z_{t}^\star - Z_{t-1}^\star}^2 
 \leq (1-\kappa^{-1})\big\|Z^{(i)}_{t-1} - Z_{t-1}^\star\big\|^2
+ 2\kappa D_0^2 \,.
 \end{align*}
  Applying this recursively gives
  \begin{align*}
 \big\|Z^{(i)}_t - Z_{t}^\star\big\|^2 &
 \leq 2 \kappa D_0^2 \sum_{s=1}^t (1-\kappa^{-1})^{t-s} 
 + (1-\kappa^{-1})^{t-1} \big\|Z^{(i)}_{1} - Z_{1}^\star\big\|^2 \\
 &\leq 2\kappa D_0^2 \sum_{s=0}^{t-1} (1-\kappa^{-1})^{s}  
 +  D_0^2 \\
 &\leq  2 \kappa^2 D_0^2 +  D_0^2 \, .
 \end{align*}
 Therefore,
\[
\what D_T = \max_{t \in [T], i,j \in [K]} \tnorm{Z_t^{(i)} - Z_t^{(j)}} 
\leq \max_{t\in[T],i,j \in [K]} \tnorm{Z_t^{(i)} - C_t^\star} 
+ \tnorm{Z_t^{(j)} - C_t^{\star}}
\leq 2\sqrt{1 + 2 \kappa^2} D_0 \,.
\]
Similarly, since the problem is unconstrained, we have $F_t(Z^\star_t) = 0$ 
(apply the definition of the solution with $ Z = Z_t^\star - F(Z_t^\star)$)
 and for any $t \in [T]$ and $i \in [k]$, 
\[
 \norm{F_t(Z_t^{(i)})}
= \norm{F_t(Z_t^{(i)}) - F_t(Z^\star_t)}
\leq L \norm{Z_t^{(i)} - Z_t^\star}\leq   L \sqrt{2 \kappa^2 + 1} D_0 \,, 
\]
proving the claim. 
\qed
\end{proof}
\smallskip

The rest of the proof is a combination of previous results, with adequate computations. 
\smallskip
\begin{proof}{Conclusion of the proof of Theorem~\ref{thm:constRegretAdaHedge}.}
First note that, by summing a geometric series, and using the fact that $\sqrt{1 + u}
  \leq 1 + u / 2$ for any $ u \in [-1, \infty]$, the correctly tuned algorithm ensures that
\begin{align*}
\sum_{t=1}^T  \big\|Z^{(k)}_t - C_t \big\| 
&\leq D_0 \sum_{t=1}^T \Big( 1 - \frac{\mu}{L} \Big)^{(1/2)\lfloor t / k \rfloor}  \\
&\leq k D_0  \sum_{i=1}^{\lfloor T / k \rfloor }\Big( 1 - \frac{\mu}{L}\Big)^{i / 2} \\
&\leq  \frac{kD_0}{1 - \sqrt{1 - \mu / L}}  
\leq \frac{kD_0}{\mu / (2L)}
=  \frac{2kLD_0}{\mu} \,.
\end{align*}
Therefore, applying Lemmas~\ref{lem:agg_ada}, \ref{lem:mix_gap_ada}, and
\ref{lem:bounded_d_and_g} specifying the sequence of comparators to $(C_t) =
  (Z_t^\star)$, 
\begin{align*}
\sum_{t=1}^{T}& \left(\dprod{ F_t(Z_t), Z_t - Z^\star_t }
- \frac{\mu}{2} \| Z_t -  Z^\star_t\|^2\right)  \\
&\leq2 \sum_{t=1}^{T} \big(\bar \ell_t - m_t\big)_+ 
+ \big(\what G_T  + \mu \what D_T \big)  
\sum_{t=1}^T  \big\|Z^{(k)}_t - Z^\star_t \big\| \\
 &\leq \frac{4(\mu \what D_T + \what G_T)^2}{\mu} \log K 
  + 2\what G_T \what D_T + \mu \what D_T^2 
+ \big(\what G_T  + \mu \what D_T \big)  
\sum_{t=1}^T  \big\|Z^{(k)}_t - Z^\star_t \big\| \\
 &\leq \frac{4 \left((2\mu + L)\sqrt{2\kappa^2 +1}\, \right)^2 D_0^2}{\mu}\log k 
+ 2 (2\mu + L)\sqrt{2\kappa^2 +1} ^2 D_0^2 
+ 2(2\mu + L)\sqrt{2\kappa^2 +1} \, k\kappa D_0^2  \\
&= 2D_0^2 (2\mu + L) \Big(2(2\kappa^2 + 1)(2 + \kappa)\log K + 1 
+ \sqrt{2\kappa^2 + 1}\kappa k \Big) \phantom{\frac{L}{\mu}}\\
&\leq 2D_0^2 (2\mu + L) \Big(2(2\kappa^2 + 1)(2 + \kappa)\log K + 1 
+ (2\kappa + 1) \kappa k\Big)\, , \phantom{\frac{L}{\mu}} 
\end{align*}  
where we used the the fact that $\sqrt{2\kappa^2 +1} \leq 2\kappa + 1$.
To conclude, recall that for any $t \in [T]$, by strong monotonicity of $F_t$, 
\[
\dprod{ F_t(Z_t), Z_t - Z^\star_t }
- \frac{\mu}{2} \| Z_t -  Z^\star_t\|^2 \geq \frac{\mu}{2} \|Z_t - Z_t^\star \|^2,
\]
so the bound above ensures that
\[
\sum_{t=1}^{T}\norm{Z_t - Z_t^\star }^2 
\leq  4D_0^2 (2 + \kappa) \Big(2(2\kappa^2 + 1)(2 + \kappa)\log K  
+ (2\kappa + 1) \kappa k + 1\Big)\, . 
\]
which is the claimed statement. 
\qed
\end{proof}


\section{On Ignoring Periodicity} 
\label{app:example_simple_periodic}
 
We illustrate one elementary yet striking consequence of ignoring the
time fluctuations, on a sequence of $1$-dimensional quadratic problems, using the
gradient descent algorithm with a fixed step size. We will see that increasing the
learning rate can move the algorithm from a diverging regime to a converging one, and
even to the fastest convergence rates. This goes against the intuitive idea that
decreasing the learning should make the sequence more stable.

For $\cZ = \R$, consider the sequence of functions $(f_t)$ defined by
$
f_{2t}(x) = x^2 / 2
$
and
$
f_{2t+1}(x) = 4 x^2
$,
which corresponds to the periodic time-varying $\mathrm{VIP}((F_t),\RR)$ with 
$
F_{2t}(x) = x
$
and
$
F_{2t+1}(x) = 8x
$, 
and with a constant solution $x^\star \equiv 0$.  We observe that $(F_t)$ defines a $2$-periodic VI problem and a VI problem with a fixed (i.e., $1$-periodic) solution.
Both of the operators $F_{2t}, F_{2t +1}$ are $1$-strongly monotone, and $8$-Lipschitz.
Consider gradient descent with step size $\eta$, which consists in selecting at all times
$t$ the sequence defined by $x_{t+1} = x_t - \eta F_{t}(x_t)$. The actions at every odd time satisfy
\[
 x_{t+2} = (1-\eta)(1-8\eta)x_t \, .
\]
Gradient descent converges to the optimum $0$ (for generic initial points) if and only if
the common ratio of the geometric sequence is less than $1$, i.e., if
$
  |1-\eta | |1 - 8\eta| < 1 \,.
$
Moreover, when the sequence converges, the common ratio dictates its speed of convergence
to the optimum: the closer it is to $0$, the faster the convergence. In particular, when
tuned with $\eta=1$, gradient descent converges to the optimum in a single step. On the
other hand, there is a range of smaller step sizes for which the sequence diverges. For
example, for $\eta = 1/2$, the common ratio is $3/2 >1$, and the iterates diverge.

 \section{Implementation details on the bifurcation diagram.}
\label{app:bifurcation_diagram}

The bifurcation diagram was obtained by running gradient descent on the sequence of 
functions defined in \eqref{eq:exp} for different values of $\eta$. 
Precisely, we ran GD from the initial point $x_0 = -0.1$ for $2000$ time steps, for a
linearly spaced grid of $3000$ values between $0$ and $8$, that is, $\eta \in \{8 i / 3000
\, : i \in \{1, \dots,3000\}\}$.

The bifurcation diagram represents the asymptotic accumulation points of the sequence of
iterates of GD: for each value of $\eta$, the points in the diagram with $x$-coordinate
$\eta$ represent the limit points of GD. Accordingly, we split the interval of values
$[-10, 10]$ into $1000$ cells and shaded all the cells that were visited between
time-steps $t = 1000$ and the end time $t = 2000$.
The grey areas correspond to values of $\eta$ that led to divergence. In practice, we 
considered that GD had diverged as soon as one iterate had a norm larger than $x_{\max} =
1000$.

All simulations were run on a standard laptop, using   \verb|python|.

 \end{document}